\let\oldvec\vec
\documentclass[natbib]{svjour3}
\let\vec\oldvec

\usepackage{amsmath,graphicx}
\usepackage{mathrsfs}
\usepackage{listings}

\usepackage{tikz}
\usepackage{xspace}
\usepackage[hyphens]{url}
\usepackage{hyperref}

\usepackage[inline]{enumitem}

\usepackage[T1]{fontenc} 
\usepackage{lmodern} 

\hypersetup{
  pdftitle					={},
	pdfauthor					={},
  colorlinks				=true,
  citecolor					={purple},
  linkcolor 				={blue},
  bookmarksnumbered	=true
}

\usetikzlibrary{arrows} 
\usetikzlibrary{automata}
\usetikzlibrary{shapes.geometric}

\renewcommand{\cite}[1]{[\citet{#1}]}
\usepackage[framemethod=tikz]{mdframed}
\usepackage{subcaption}
\usepackage{etoolbox}
\usepackage{stfloats}

\usepackage{flushend} 

\usepackage[absolute]{textpos}

\newcommand{\seclabel}[1]{\label{sec:#1}}
\newcommand{\secref}[1]{Section~\ref{sec:#1}}

\newcommand{\thmlabel}[1]{\label{thm:#1}}
\newcommand{\thmref}[1]{Theorem~\ref{thm:#1}}

\newcommand{\rmlabel}[1]{\label{rm:#1}}
\newcommand{\rmref}[1]{Remark~\ref{rm:#1}}

\newcommand{\deflabel}[1]{\label{def:#1}}
\newcommand{\defref}[1]{Definition~\ref{def:#1}}

\newcommand{\lemlabel}[1]{\label{lem:#1}}
\newcommand{\lemref}[1]{Lemma~\ref{lem:#1}}

\newcommand{\eqlabel}[1]{\label{eq:#1}}
\renewcommand{\eqref}[1]{(\ref{eq:#1})}
\newcommand{\applabel}[1]{\label{app:#1}}
\newcommand{\appref}[1]{Appendix~\ref{app:#1}}

\title{Order-Reduction Abstractions for Safety Verification of High-Dimensional Linear Systems}

\titlerunning{Order-Reduction Abstractions for Safety Verification of High-Dimensional Systems}        

\author{Hoang-Dung Tran, Luan Viet Nguyen, Weiming Xiang, Taylor T. Johnson}
\institute{University of Texas at Arlington, USA}

\date{Received: February 15, 2016 / Accepted: date}

\makeatletter
\def\blfootnote{\xdef\@thefnmark{}\@footnotetext}
\makeatother

\newcommand{\loc}{\ell}

\newcommand{\upd}{\mu}

\newcommand{\constr}[1][]{\varphi\ifthenelse{\equal{#1}{}}{}{_#1}}

\newcommand{\ofsane}[1]{\ifthenelse{\equal{#1}{(}}{\errmessage{Hups, '(' as of-macro parameter?}}{}}

\newcommand{\varsofX}[1]{\ofsane{#1}\mathit{vars}(#1)}

\newcommand{\varsofconstr}[1][]{\varsofX{\upd}}

\newcommand{\iloc}[1][]{\loc_{\mathit{ini\ifthenelse{\equal{#1}{}}{}{,#1}}}}
\newcommand{\nstloc}[1][]{\loc_{\mathit{nst\ifthenelse{\equal{#1}{}}{}{,#1}}}}

\newcommand{\prosaDONE}[1]{}

\newcommand{\todo}[1]{}
\newcommand{\stan}[1]{}
\newcommand{\taylor}[1]{}
\newcommand{\hoang}[1]{}
\newcommand{\sergiy}[1]{}

\newcommand{\figlabel}[1]{\label{fig:#1}}
\newcommand{\figref}[1]{Figure~\ref{fig:#1}}

\newcommand{\tablabel}[1]{\label{tab:#1}}
\newcommand{\tabref}[1]{Table~\ref{tab:#1}}

\usepackage[
  pass,
]{geometry}

\definecolor{light-gray}{gray}{0.95}

\def\loc{m}

\definecolor{gri}{gray}{.75}

\newcommand{\deq}{\mathrel{\stackrel{\scriptscriptstyle\Delta}{=}}}

\newcommand{\norm}[1]{\left\|#1\right\|}

\newcommand{\commentTaylor}[1]{} 

\usepackage{xcolor}
\usepackage{makecell}
\usepackage{adjustbox}
\RequirePackage{tikz}
\usetikzlibrary{arrows,automata,shapes,calc,through,decorations.pathmorphing,decorations.fractals,chains,shapes.multipart}
\usepackage{amssymb}
\usepackage{graphicx}
\usepackage{amsmath}
\usepackage{amsfonts}
\usepackage{multirow}
\usepackage{tabulary}

\usepackage[hyphens]{url}

\usepackage{balance}

\newfont{\mycrnotice}{ptmr8t at 7pt}
\newfont{\myconfname}{ptmri8t at 7pt}

\newcommand{\deds}{{Discrete Event Dynamic Systems.}}
\journalname{\deds}

\begin{document}

\maketitle

\begin{abstract}

Order-reduction is a standard automated approximation technique for computer-aided design, analysis, and simulation of many classes of systems, from circuits to buildings. For a given system, these methods produce a reduced-order system where the dimension of the state-space is smaller, while attempting to preserve behaviors similar to those of the full-order original system. To be used as a sound abstraction for formal verification, a measure of the similarity of behavior must be formalized and computed, which we develop in a computational way for a class of linear systems and periodically-switched systems as the main contributions of this paper. We have implemented the order-reduction as a sound abstraction process through a source-to-source model transformation in the HyST tool and use SpaceEx to compute sets of reachable states to verify properties of the full-order system through analysis of the reduced-order system. Our experimental results suggest systems with on the order of a thousand state variables can be reduced to systems with tens of state variables such that the order-reduction overapproximation error is small enough to prove or disprove safety properties of interest using current reachability analysis tools. Our results illustrate this approach is effective to alleviate the state-space explosion problem for verification of high-dimensional linear systems.


\end{abstract}

\section{Introduction}
\seclabel{intro}

The state-space explosion problem is a fundamental challenge in model checking and automated formal verification that has received significant attention from the verification community.
Roughly, the state-space explosion problem is that the size of the state-space of systems scales exponentially or combinatorially with their dimensionality, which in turns causes formal computational analyses of these systems to scale similarly.
Among many solutions, abstractions based on the concepts of exact and approximate simulation and bisimulation relations are effective approaches to obtain smaller state spaces by abstracting away information that is not needed in the verification process.
Such abstractions have been applied broadly to simplify the controller synthesis and safety verification process of complex systems.
Applications of these abstractions can be found in many fields such as embedded systems~\cite{henzinger2006embedded}, biological systems~\cite{danos2004formal,regev2004bioambients,asarin2004abstraction}, continuous and hybrid systems models of cyber-physical systems (CPS)~\cite{alur2000discrete, belta2005discrete, girard2007approximation, girard2008approximate}, and stochastic systems~\cite{wang2015hscc}.

Model reduction techniques have been developed and applied widely in controls~\cite{Antoulas01asurvey}, but are typically approximations and not sound abstractions that may be used in formal verification.
From a high-dimensional (``full-order'') original system, model reduction can obtain automatically a simplified (``reduced-order'') system with lower-dimensionality that is computationally easier to, for example, analyze, design controllers for, and simulate.

A key difference between model reduction and abstraction relies on dealing with the system's initial condition and inputs. In model reduction, the inputs' values remain the same and the initial set of system states, which is an important factor in verification with reachability analysis, is usually assumed to be the zero set.
Thus, order reduction may not be sound as it is an approximation that may not have bounded errors and may be subject to numerical errors, while a guaranteed error bound (``conservative approximation'') is necessary to define a sound abstraction for verification.
In contrast, the initial condition is always taken into account and the inputs' values may change in the context of bisimulation-based abstraction.

In order to be able to use model reduction as a sound abstraction for formal verification, we need to consider both the initial conditions and inputs and then present additional reasoning to derive error bounds for how far off the executions of the reduced-order system may be from those of the full-order system.
Formalizing these issues and addressing them is the main objective of this paper, which we then use to derive an automated order-reduction abstraction that is sound, and use it to verify high-dimensional (with upwards of hundreds to thousands of state variables) continuous and periodically switched systems.

\subsection{Related Work}
Exact bisimulation relation-based abstractions for safety verification and controller synthesis have been investigated widely in the last decade~\cite{pappas2003bisimilar, van2004equivalence, tanner2003abstractions, tabuada2004bisimilar}.
In this context, the outputs of the the abstract system capture exactly the outputs of the original system.
As pointed out in~\cite{girard2008approximate, girard2007approximate}, the term of ``exact'' is not adequate when dealing with continuous and hybrid systems observed over real numbers since there may be numerical errors in observation, noise, among other nonidealities. 
To obtain an abstraction that guarantees more robust relationships between systems, approximate bisimulation relations have been proposed and studied extensively in recent years~\cite{girard2005approximate, girard2008approximate, girard2007approximate, julius2006approximate, girard2006approximate, islam2015tcs}.
The main advantage of such approximate relations is that they allow a bounded error $\delta$ which describes how far off the executions of the abstraction may be from those of the original system. 
Then, verifying whether the executions of the original system reach an unsafe region $U$ can be turned out to verify whether the executions of the abstraction (with much lower dimension) reach the $\delta$-neighborhood of the unsafe region $U$. 
Thus, finding an efficient way to determine a tight bound of the error becomes an essential task for this approach.    
In particular, a computation framework has been proposed and integrated in a Matlab toolbox called Matisse to find an abstraction from the original linear system and calculate the bound of their output mismatch~\cite{girard2007approximate}. Note that in~\cite{girard2007approximate}, the error bound is called as a precision. 
The proposed method shows a great benefit when it can deal both stable and unstable systems.
In this framework, computing the error bound is basically based on solving a set of linear matrix inequality (LMI) and optimization problem on the sets of initial states and inputs. 
The computation complexity increases polynomially along with the size of the system.
In addition, in some cases, due to the ill-condition of some matrices in computation process, the error bound computed may be very conservative and thus may produce an abstraction that is not useful for verification.

Model reduction techniques have been applied for formal verification of continuous and hybrid system~\cite{han2004reachability,han2005formal}.
These techniques rely on output reach sets, which combines the set of reachable states and an observation matrix.
This concept is useful in safety verification because, for a given system, we are usually interested in the safety requirements of some specific important states or their combinations which can be declared as the outputs of the system.
Particularly, the authors use a reduced-order model and its output error bound compared with full-order model to overapproximate the output reach set of the original system~\cite{han2004reachability}.
Thus, determining a tight bound of the error is essential.
Intuitively, the error between full-order model and its reduced-order model is composed of two separate errors.
The first error corresponds to the zero-input response (i.e, there is no control input) and the second error corresponds to the zero-state response (i.e., the initial state of the system is zero).
The authors use simulation to determine the bounds of these errors before combining them as a total bound.
The first error bound is determined by simulating the the full-order system and the reduced model from each vertice of a polyhedral initial set of states.
The advantages of simulation is it can derive tight bounds of these errors.     
The drawback of using simulation is the number of simulations increases exponentially with the dimension of the polyhedron, since the number of vertices of a polyhedron increases exponential with the dimensionality (for example, an $n$-dimensional hyperbox has $2^n$ vertices).
Thus, it may be infeasible to perform enough simulations for a high-dimensional system.

Reachability analysis of large-scale affine systems has also been investigated with Krylov subspace approximation methods to deal with state-space explosion~\cite{han2006reachability}.
However, this approach requires the input to the system to be constant, while in contrast, in our work, we consider a more general class of systems with varying inputs.

%
\paragraph*{Contributions and Organization.~~}

In this paper, we develop the order-reduction abstraction for safety verification of high-dimensional linear systems. 
The main contributions of this paper are:
(a) a computationally efficient method to derive an output abstraction from high-dimensional linear systems with an error bound for each element of the outputs, where this error is essentially the sum of two separate errors caused by the initial set of states and the control inputs; 
(b) establishing soundness of using this output abstraction to verify the safety requirements of the original system with a significantly lower computation cost;
(c) an extension of these results to a class of periodically switched linear systems; and
(d) the implementation of the methods as a model transformation pass within the HyST model transformation tool~\cite{bak2015hscc}, along with a thorough evaluation comparing our approach to similar existing order-reduction and approximate bisimulation-based abstraction methods.

Our computational framework has been tested and compared in detail with similar results through a set of benchmarks~\cite{girard2007approximate,han2004reachability}.
Our empirical evaluation illustrates that our method not only works efficiently for small and medium-dimensional systems (several to less than a hundred dimensions) as existing methods~\cite{girard2007approximate,han2004reachability}, but also can be applied to high-dimensional systems (with a hundred to a thousand dimensions) where the existing methods are infeasible to apply due to either computational complexity or finding overly conservative error bounds.
The error bound value and computation time of these different methods have been compared and discussed in our paper to show the advantages and tradeoffs of our approach.

Besides improving the computational framework, we also establish soundness of our method using the output abstraction to verify the safety specifications of the original full-order system using an approximate bisimulation relation argument.
In~\cite{girard2007approximate}, the authors use the general concept of set neighborhood to transform the safety specification of the original system.
However, in some cases when the safe and unsafe regions are described by polytopes or ellipsoids that are often used, for example in SpaceEx~\cite{frehse2011spaceex} and the Ellipsoid Toolbox~\cite{kurzhanskiy2006ellipsoidal}, a more precise transformed safety specification can be derived by our element-to-element approach.
The transformed safety specification need to satisfy the following safety relation property: (1) if the output abstraction is safe (i.e., it satisfies the transformed safety specification), then the original system is safe, and (2) if the abstraction is unsafe (it does not satisfy the transformed safety specification), then the original system is unsafe.
Since we verify safety using the output abstraction, the computation cost of the verification process is significantly reduced.
Moreover, our approach is very useful for verifying safety of high-dimensional systems that the existing verification tools may not successfully analyze directly.
This improvement is shown through our evaluation of computation complexity of safety verification for the original full-order system and its different output abstractions (\secref{casestudies}, \tabref{computation}).
Our method has been implemented as a source-to-source model transformation in the HyST tool~\cite{bak2015hscc}, which makes it easy to combine different verification tools, such as SpaceEx~\cite{frehse2011spaceex}, Flow*~\cite{chen2013flow} and dReach~\cite{kong2015dreach} to verify safety property of high-dimensional linear systems.

The remainder of the paper is organized as follows.
\secref{concept} gives definitions of output reach set, output abstraction, safety specification, safety verification problem and safety specification transformation for a class of linear time invariant (LTI) systems.
\secref{optainoutputabstraction} presents methods to find output abstractions of the LTI systems using the balanced truncation model reduction method.
\secref{verificationusingoutputabstraction} discusses how to verify safety properties for a full-order LTI system using its output abstraction.
\secref{abstraction_hybrid} extends the results to a class of periodically switched systems in which the state of the system is re-initialized at every switching instance.
\secref{casestudies} describes our implementation of the method in a prototype tool, and presents a number of examples to illustrate and evaluate the benefits of our method.

\section{Preliminaries}
\seclabel{concept}
In this section, we introduce definitions used throughout the paper including output reach sets~\cite{han2004reachability}, output abstractions, safety specifications, the safety verification problem, and safety specification transformation.
\begin{definition}
An \emph{$n$-dimensional LTI system} is denoted $M_{n}(y|\{x,u\})\langle A,B,C \rangle$ (written in short as $M_n$).
It has the following dynamic equations:
\begin{equation*}
\eqlabel{full-order}
\begin{split}
\dot{x}(t) &= A x(t) + B u(t) \\
     y(t)  &= C x(t),
\end{split}
\end{equation*}
where $x(t) \in \mathbb{R}^{n}$ is the \emph{system state}, $y(t) \in \mathbb{R}^{p}$ is the \emph{system output}, $u(t)$ is the \emph{control input}, $A \in \mathbb{R}^{n\times n}$, $B \in \mathbb{R}^{n\times m}$, and $C \in \mathbb{R}^{p\times n}$.
\end{definition}

The initial set of states of $M_{n}$ is denoted by $X_{0}(M_{n}) \subseteq \mathbb{R}^n$, and we write initial conditions as $x(0) \in X_{0}(M_n)$.
The set of control inputs of $M_{n}$ is $\textsl{U} \subseteq \mathbb{R}^m$, and we write particular controls as $u(t) \in \textsl{U}$. 
The state of $M_{n}$ is updated from state $x$ to the new state $x'$ over intervals of real time according to the linear differential equation $Ax + Bu$, and the \emph{behaviors} of the system is defined in this paper as the trajectories of the output $y(t)$ over intervals of real time.
\commentTaylor{This should really be defined more nicely (e.g., as a sequence if you like, or as executions as in standard hybrid automata papers), but whatever at this point.}

Next, we present the output reach set defined in related approaches using order reduction as a sound abstraction~\cite{han2004reachability}.
\begin{definition}{Output Reach Set~\cite{han2004reachability}.}%
Given an LTI $M_n$, a set of control inputs $\textsl{U}$, and an initial set $X_{0}(M_n)$, the \emph{output reach set at a time instant $t$} is:
\begin{equation*}
\begin{split}
R_t(M_n) &\deq \{ y(t,u,x_0) | y(t,u,x_0) = Ce^{At}x_0 + \int_{t_0}^{t}Ce^{A(t-\tau)}Bu(\tau)d\tau\},\\
&\mbox{ where}~x_0 \in X_{0}(M_n) \mbox{ and } u(t) \in \textsl{U}.
\end{split}
\end{equation*}

The \emph{output reach set over an interval of time $[t_0,t_f]$ for $t_0 \leq t_f$} is:
\begin{equation*}
\begin{split}
R_{[t_0,t_f]}(M_n) \deq \bigcup_{t\in [t_0,t_f]} R_t(M_n).
\end{split}
\end{equation*}
\end{definition}

In the remainder of the paper, we suppose $t_0 = 0$.
Next, we define an output abstraction, which is a formalization of the reduced-order system that will be used to verify properties of the full-order system.
\begin{definition}{Output Abstraction}.%
\deflabel{output_abstraction}%
The $k$-dimensional LTI system $M_{k}^{\delta}$, $p < k \leq n$ described by:
\commentTaylor{I changed to $p < k$. If you want to leave it as $0 < k$, which could be okay by also say projecting away an output variable, you need to add a remark to clarify, as otherwise it's confusing as everything else talks about doing verification with respect to outputs, so if you project away all the outputs, what happens? Also, it's a bit unclear what happens if the outputs are projected away, I guess this is handled in the safety transformation, but that definitely deserves a remark as it's a subtlety.}
\begin{equation*}
\begin{split}
\dot{x}_r(t) &= A_r x_r(t) + B_r u(t), \\
     y_r(t)  &= C_r x_r(t),
\end{split} 
\end{equation*}
where $x_r(t) \in \mathbb{R}^{k}$, $y_r(t) \in \mathbb{R}^{p}$, $A_r \in \mathbb{R}^{k\times k}$, $B_r \in \mathbb{R}^{k\times m}$, $C_r \in \mathbb{R}^{p\times k}$, is called a \emph{$k$-dimensional output abstraction of $M_{n}$} if, for the \emph{error bound} $\delta = [\delta_1, \delta_2, \ldots, \delta_p]^T$, where each $\delta_i$ is a finite positive real, we have:
\begin{enumerate}
\item $\forall x(0) \in X_{0}(M_{n})$ and $u \in \textsl{U}$, $\exists x_r(0) \in X_{0}(M_{k}^{\delta})$ such that, $\forall t \geq 0$, $\norm{y^{i}(t)-y_r^{i}(t)} \leq \delta_i$, $1 \leq i \leq p$.
\end{enumerate}
where $y^{i}(t)$ is the $i^{th}$ component of the output $y$ at time $t$, and $\norm{\cdot}$ denotes the Euclidean norm.
\commentTaylor{Hoang-dung: please confirm and make sure this notation (e.g., $y^{i}(t)$ meaning) is consistent with the rest of the paper}
\end{definition}
If we can find an output abstraction $M^{\delta}_k$, then its behaviors will approximate within $\delta$ the behaviors of the full-order system $M_n$ for all time.

\begin{definition}{Safety Specification.}
A \emph{safety specification} $S(M_n)$ of an LTI system $M_n$ formalizes the safety requirements for the output $y$ of $M_n$, and is a predicate over the output $y$ of $M_n$.
Formally, $S(M_n) \subseteq \mathbb{R}^{p}$.
\commentTaylor{Any assumptions on this set? Can this be a complex set, e.g., semialgebraic, topologically complex, etc.? If so, want to say}
\end{definition}

\begin{definition}{Safety Verification.}
The \emph{time-bounded safety verification problem} is to verify whether the system $M_n$ satisfies a safety specification $S(M_n)$ over an \emph{interval of time}.
Whether $M_n$ is safe or unsafe is defined over \emph{an interval of time} $[0,t_f]$, which is described formally in terms of the output reach set as:
\commentTaylor{Inconsistency: earlier you said only you would consider $t_0 = 0$, but now you have $t_0$ again. Also earlier you defined $t_f$ using $t_s$, keep it consistent. I fixed it.}
\begin{equation*}
\begin{split}
&R_{[0,t_f]}(M_n) \cap \neg S(M_n) = \emptyset \Leftrightarrow M_n \vDash S(M_n) , \\
&R_{[0,t_f]}(M_n) \cap \neg S(M_n) \neq \emptyset \Leftrightarrow M_n \nvDash S(M_n). 
\end{split}
\end{equation*}
\commentTaylor{In the future, set up macros for all your notation so it's easy to change. This is the reason it's important to do this, as we really need to change these as the definition of safety is only over an interval of time, not all time.}
In the remainder of the paper, we will assume $t_f$ is finite and focus on time-bounded safety verification, albeit the general framework we develop and error bounds we derive are applicable to time-unbounded verification where $t_f \rightarrow \infty$.
Under this assumption, we fix $t_f$ to some positive real.
If $M_n$ satisfies $S(M_n)$, then it is \emph{safe} and we write $M_n \vDash S(M_n)$.
If $M_n$ does not satisfy $S(M_n)$, then it is \emph{unsafe} and we write $M_n \nvDash S(M_n)$.
\end{definition}

\begin{definition}{Safety Specification Transformation.}
The \emph{safety specification transformation} is the process of finding the corresponding safety (or dually, unsafe) specification for the output abstraction $M_{k}^{\delta}$ denoted by $S(M_{k}^{\delta}) \in \mathbb{R}^{p}$ (and dually $U(M_{k}^{\delta}) \in \mathbb{R}^{p}$) from the safety specification $S(M_n)$ of the full-order system $M_n$ to guarantee the safety relation defined by:
\begin{equation}\eqlabel{safety_relation}
\begin{split}
&R_{[0,t_f]}(M_k^{\delta}) \cap \neg  S(M_k^{\delta}) = \emptyset \Rightarrow  M_n \vDash S(M_n), \\ 
&R_{[0,t_f]}(M_k^{\delta}) \cap   U(M_k^{\delta}) \neq \emptyset \Rightarrow M_n \nvDash S(M_n). 
\end{split}
\end{equation}
\end{definition}

\section{Output Abstractions from Balanced Truncation Reduction}
\seclabel{optainoutputabstraction}

The balanced truncation model reduction is an effective method to find reduced models for large scale systems.
Balanced truncation is based on Singular Value Decomposition (SVD)~\cite{moore1981principal} and uses a balanced projection to transform a system to an equivalent \emph{balanced system} where the states are arranged in descending degrees of their controllability and observability.
Informally, the degrees of controllability and observability are measures to check how controllable and observable a given system is.
For further details, we refer readers to~\cite{moore1981principal,silverman1967controllability}.
The $k$-order reduced model is then obtained by selecting the first $k$ states in the state vector and truncating (i.e., projecting away or eliminating) the other $n-k$ states.
The process of determining the reduced-order model's matrices is well-known, and it is briefly described here.
We refer readers to~\cite{moore1981principal} for further details.

\subsection{Order Reduction with Balanced Truncation Method}
\seclabel{reduction}
For an LTI system $M_{n}$, the controllability gramian $W_{c}$ and observability gramian $W_o$ of $M_n$ are the solutions of the following Lyapunov equations,%
\begin{align*}
AW_{c}+W_{c}A^{T}+BB^{T} &= 0\\
A^{T}W_{o}+W_{o}A+C^{T}C &= 0.
\end{align*}

It should be noticed that $W_{c}$ and $W_{o}$ are symmetric and positive definite.
The Hankel singular value $\sigma_{i}$ is defined as the square root of each eigenvalue $\lambda_i$ of $W_{c}W_{o}$,%
\[ \sigma_{i} = (\lambda_{i}(W_{c}W_{o})^{\frac{1}{2}}). \]
	
The first step in balanced model reduction method is to implement a balanced transformation $\tilde{x}(t) = Hx(t),~H \in \mathbb{R}^{n}$ to transform $M_n$ to an equivalent balanced system $\widetilde{M}_n$, where the controllability and observability gramian $\widetilde{W_{c}},\widetilde{W_{o}}$ satisfy:%
\[\widetilde{W_{c}} = \widetilde{W_{o}} = \Sigma = \left(
  \begin{array}{cccc}
    \sigma_{1} &  &  & \\
               & \sigma_{2}  &  & \\
     &  & \ddots & \\
	   &  &  &  \sigma_{n} \\
  \end{array}
\right),
\]\\
for $\sigma_{1} \geq \sigma_{2} \geq \sigma_{3} \geq \cdots \geq \sigma_{n-1} \geq \sigma_{n}$.
The transformation matrix $H$ can be computed as follows. 

Since $W_{c}$ is symmetric and positive definite, we can factor it as $W_{c} = GG^{T}$, where $G$ is invertible.
There exists an orthogonal transformation $K$, (i.e. $KK^{T} = I$, where $I$ is an identity matrix) such that $G^{T}W_{o}G = K\Sigma^{2}K^{T}$.
Then, the transformation matrix $H$ is defined as:%
\[H = \Sigma^{\frac{1}{2}}K^{T}G^{-1}.\] 

Applying transformation to the system $M_n$, we have the equivalent balanced system $\widetilde{M}_n$ with:%
\begin{equation*}
\begin{split}
\dot{\tilde{x}}(t) &= \tilde{A}\tilde{x}(t)+\tilde{B}u(t) \\
     y(t)  &= \tilde{C}\tilde{x}(t),
\end{split}
\end{equation*}
where $\tilde{A} = H A H^{-1}$, $\tilde{B}=H B$ and $\tilde{C} = CH^{-1}$.
The matrices of $\widetilde{M}_n$ can be partitioned as:%
\[\tilde{A} = \left(
  \begin{array}{cc}
    \tilde{A}_{11} & \tilde{A}_{12} \\
    \tilde{A}_{21} & \tilde{A}_{22} \\
  \end{array}
\right), \tilde{B} = \left(
  \begin{array}{c}
    \tilde{B}_{1} \\
    \tilde{B}_{2} \\
  \end{array}
\right), \tilde{C} = \left(
  \begin{array}{cc}
    \tilde{C}_{1} & \tilde{C}_{2} \\
  \end{array}
\right),\] \\
where $\tilde{A}_{11}\in \mathbb{R}^{k\times k}$,  $\tilde{B}_1 \in \mathbb{R}^{k\times m}$, and $\tilde{C}_1 \in \mathbb{R}^{p\times k}$, and the other matrices are of appropriate dimensionality. 
Finally, the $k$-dimensional reduced system $M_{k}$ of $M_n$ is defined as:
\begin{equation}\eqlabel{red_sys}
\begin{split}
\dot{x}_r(t) &= A_r x_r(t) + B_r u(t) \\
     y_r(t)  &= C_r x_r(t).
\end{split}
\end{equation}
where $A_r = \tilde{A}_{11}$, $B_r = \tilde{B}_1$ and $C_r = \tilde{C}_1$. 
The initial set of $M_k$ is $X_{0}(M_k) = \{SHx_0 | x_0 \in X_{0}(M_{n}) \}$, where $S = \left(
\begin{array}{cc}
    I_{k \times k} & 0_{k \times (n-k)} \\
  \end{array}
\right)$.

The balanced truncation method obtains the system matrices of the $k$-dimensional reduced system.
Next, we investigate the error between the outputs of the full-order system and its $k$-dimensional reduced system.

\subsection{Determining the Error Bound $\delta$}
\seclabel{error_bound}
The solution of a LTI system can be decomposed into two parts.
The first part corresponds to zero control input (i.e. $u(t) = 0$ for all $t$) and the second part corresponds to zero initial state (i.e., $x_0 = 0$). Note that $x_0 \equiv x(t = 0)$. 
The solutions of $M_n$ and $M_k$ are given as follows:
\commentTaylor{These really need some clarification with time, please double check.}
\begin{equation*}
\begin{split}
&y(t)= y_0 + y_{u}(t),~y_{r}(t) = y_{r_{0}}(t) + y_{r_{u}}(t) \\ 
&y_0(t) = Ce^{At}x_{0},~y_{u}(t)= \int_{0}^{t}Ce^{A(t-\tau)}Bu(\tau)d\tau \\
&y_{r_{0}}(t)= C_{r}e^{A_{r}t}x_{r_{0}},~y_{r_{u}}(t) = \int_{0}^{t}C_{r}e^{A_{r}(t-\tau)}B_{r}u(\tau)d\tau \\
&x_0(t) \in X_0(M_n),~x_{r_0}(t)\in X_0(M_k),~u(\tau) \in \textsl{U}.\\ 
\end{split}
\eqlabel{solution}
\end{equation*}
\commentTaylor{Consistency: sometimes you seem to use $x_0$ for initial conditions and sometimes you use $x(0)$. Please resolve.}

The error between the full-order system $M_n$ and its $k$-dimensional reduced system $M_k$ at time $t$ is given as follows:
\begin{equation}\eqlabel{error}
\begin{split}
e(t)~ &= y(t) - y_r(t) = e_{1}(t) + e_{2}(t), \mathrm{ where} \\ 
e_{1}(t) &= y_{0}(t) - y_{r_{0}}(t),~\mathrm{ and }~e_{2}(t) = y_u(t) - y_{r_{u}}(t).
\end{split}
\end{equation}

The error $e_1$ relates to the zero input state responses of the full-order system and the output abstraction; that is, the responses are only caused by the initial set of states.
The error $e_2$ relates to the zero state responses; that is, the responses are only caused by the control inputs.
Note that $e_1$ and $e_2$ are both time varying and $e_1(t),~e_2(t) \in \mathbb{R}^p$ where $p$ is the output dimensionality.
\commentTaylor{I don't understand a hundred percent how the error is defined. I guess you mean these are vectors? Typically errors are normed differences, so you need to clarify. E.g., it seems $e_1(t) \in \mathbb{R}^p$ and $e_2(t) \in \mathbb{R}^p$, where $p$ is the output dimensionality. If that's the case, it's worth a remark.}

To obtain our main results in computing the error bounds, in the rest of this paper, we consider the $(n+k)$-dimensional augmented system as follows, 
\begin{equation*}
\begin{split}
\dot{\bar{x}} &= \bar{A}\bar{x} + \bar{B}u = \begin{pmatrix}  \tilde{A} & 0 \\ 0 & A_r \\\end{pmatrix} \bar{x} + \begin{pmatrix} \tilde{B} \\ B_r \end{pmatrix} u, \\ 
    {\bar{y}} &= \bar{C}\bar{x} = \begin{pmatrix} \tilde{C} & -C_r \end{pmatrix} \bar{x},
\end{split}
\end{equation*}
where $\bar{x} = \begin{pmatrix} Hx & SHx \end{pmatrix}^T$.

It is easy to see that the output of the augmented system is the error between the $n$-dimension full-order system and its $k$-dimensional reduced system.
Thus, determining the error bound $\delta$ is equivalent to determining the bounds of the augmented system's outputs in which the bound of $e_1$ corresponds to the zero input response, while the bound of $e_2$ relates to the zero state response of the augmented system.  

A theoretical bound of $e_1$ can be given with the following theorem.
\commentTaylor{You said ``the theoretical bound''. That would imply it's tight. If it's not, using ``a'' is safer.}
\begin{theorem}%
\thmlabel{e1_bound}%
Let $\bar{x}_0 = \begin{pmatrix} Hx_0 & SHx_0 \end{pmatrix}^T $, then the error $e_1$ between the full-order system $M_n$ and its $k$-dimensional reduced system $M_k$ satisfies the following inequality for all $t \in \mathbb{R}_{\geq 0}$:
\begin{equation*}
\left\|e_1^i(t)\right\| \leq  \lambda_{max}(\bar{C}_i^T\bar{C}_i)  \cdot \sup_{x_0 \in X_0}  \left\|\bar{x}_{0}\right\|,~ 1 \leq i \leq p, 
\end{equation*}
where $\mathbb{R}_{\geq 0}$ is the set of non-negative real numbers, $\bar{C}_i$ is the row $i$ of the matrix $\bar{C}$, and $e_1^i(t)$ is the $i^{th}$ element of vector $e_1(t)$.
\commentTaylor{Where is time $t$? If this for all time? If so, say it. If not, you have to clarify. }
\commentTaylor{Some consistency: earlier you refer to $y(i)$ as the $i^{th}$ element of a vector $y$, while now you are using subscripts for $C_i$. I realize this is now a matrix, but still, want to keep your notation consistent throughout, so please resolve, as it's confusing using $e_1(i)$ to refer to the $i^{th}$ entry of $e_1$, but using $C_i$ to refer to the $i^{th}$ row of $C$.}
\end{theorem}

The proof of~\thmref{e1_bound} is given in~\appref{Appendix_alg}. 
 
Although the computation cost of the theoretical bound of $e_1$ is small since it only relates to determining $\sup_{x_0 \in X_0} \left\| \bar{x}_{0} \right\|$, the result may be very conservative in the case that the initial set of states is far from the zero point. 
Thus, as can be seen in~\secref{casestudies}, it is efficient to use~\thmref{e1_bound} to compute the bound of $e_1$ if the initial set of states is close to zero point. 

To reduce the conservativeness of~\thmref{e1_bound}, we also propose an optimization method to compute a tighter bound of $e_1$. 
Nevertheless, the computation cost of the optimization method is larger than using~\thmref{e1_bound}. 
This is the tradeoff between obtaining an accuracy bound of $e_1$ and improving the computation time.  
An optimization method is given in the following theorem. 
\begin{theorem}%
\thmlabel{e1_bound_opt}%
Let $\bar{x}_0 = \begin{pmatrix} Hx_0 & SHx_0 \end{pmatrix}^T $ and $P_0 > 0$ is the solution of the following optimization problem: 
\begin{equation*}
\begin{split}
&P_0 = min(trace(P))~subject~to\\
&P > 0,~\bar{A}^TP+PA <0,~\bar{C}_i^T\bar{C}_i \leq P \\ 
\end{split}
\end{equation*} 
where $\bar{C}_i$ is the row $i$ of the matrix $\bar{C}$.
Then, the error $e_1$ between the full-order system $M_n$ and its $k$-dimensional reduced system $M_k$ satisfies the following inequality for all $t \in \mathbb{R}_{\geq 0}$:
\begin{equation*}
\left\|e_1^i(t) \right\| \leq \sup_{x_0 \in X_0} \sqrt{\bar{x}_{0}^TP_0\bar{x}_0},~ 1 \leq i \leq p, 
\end{equation*}
where $e_1^i(t)$ is the $i^{th}$ element of vector $e_1(t)$.
\commentTaylor{Again, for all time?}
\end{theorem}

The proof of~\thmref{e1_bound_opt} is given in~\appref{Appendix_alg}.

Now, we consider how to determine the bound of $e_2$, which corresponds to the zero state responses of the augmented system.
The theoretical bound of $e_2$ is obtained in the following theorem by exploiting the concept of bounded input bounded output stability (BIBO) and the $L_1$ error bound in the impulse response of balanced truncation model reduction~\cite{obinata2012model}. 

\begin{theorem}%
\thmlabel{e2_bound}%
The error $e_2$ between the full-order system $M_n$ and its $k$-dimensional reduced system $M_k$ satisfies the following inequality for all $t \in \mathbb{R}_{\geq 0}$:
\begin{equation*}
\left\|e_2^i(t) \right\| \leq  (2\sum_{j=k+1}^{n}(2j-1)\sigma_j) \cdot \left\| u \right\|_{\infty},~ 1 \leq i \leq p, 
\eqlabel{output_bound}
\end{equation*}
where $e_2^i(t)$ is the $i^{th}$ element of vector $e_2(t)$.
\commentTaylor{for all time?}
\end{theorem}

The proof of~\thmref{e2_bound} is given in~\appref{Appendix_alg}. 

\begin{remark}
\rmlabel{e2_rm}
As can be observed from~\thmref{e2_bound}, the theoretical bound of $e_2$ depends on the singular value $\sigma_j,~k+1\leq j \leq n$. 
Therefore, in the case of the singular values are large, the theoretical bound of $e_2$ becomes large and may be not useful.
Moreover, ~\thmref{e2_bound} derives the same error bounds $e_2^i$ for each pair $(y_u^i,~y_{r_u}^i)$ for each dimension $1\leq i \leq p$. 
Thus,~\thmref{e2_bound} may be more useful for systems that have high dimensions and small singular values. 
\end{remark}
After determining the bounds $e_1$ and $e_2$, the overall error bound $\delta = [\delta_1, \delta_2,...,\delta_p]^T$ between the outputs of $M_n$ and $M_k$ obtained from~\eqref{error} can be expressed as follows:
\begin{equation}\eqlabel{delta}
\left\|y_i(t)-y_{r,i}(t)\right\| \leq \delta_i,~1 \leq i \leq p.
\end{equation}
 where $ \delta_i= \left\|e_1^i(t)\right\| + \left\| e_2^i(t)\right\|$.
\commentTaylor{Fix time aspect here, inconsistent with no times on earlier ones.}

\begin{remark}
We note that the bound $\delta$ can be obtained using different methods in which each method has both benefits and drawbacks.
For example, in contrast to our above results, in~\cite{han2004reachability}, the authors propose a simulation-based approach to determine these error bounds.
To determine the bound of $e_1$, the author simulate the full-order system and the reduced system from each vertex in a polyhedral representation of the initial set of states.
This method gives a very tight bound of $e_1$.
The drawback is the number of simulations may explode.
For example, if the initial set is a hypercube in $100$-dimensions, we have to simulate the full-order system and its reduced system with $2^n = 2^{100}$ vertices, which is infeasible even if each simulation takes little time.
The bound of $e_2$ is determined by integrating the norm of the impulse response of the augmented system via simulation.
This method is useful since it gives a tight bound of $e_2$ with only $m$ simulations, where $m$ is the number of inputs.
In a different way without separately computing the bounds of $e_1$ and $e_2$, the error bound $\delta$ can be calculated by solving a set of LMI optimization problem on sets of initial states and inputs~\cite{girard2007approximate}.
This approach shows advantages when dealing with small and medium-dimensional systems (less than $50$ dimensions) and it works for both stable and unstable systems.
When the system dimension is large, the error bound obtained is overly conservative and may not useful.
\end{remark}

\paragraph*{Discussion.~~}
Although simulation-based methods can be used to determine the bounds of $e_1$ and $e_2$, numerical issues in simulation may lead to unexpected results (unsound results) which are smaller than the actual error bounds.
Let us clarify the problem first and then propose a technique under an assumption to make the result obtained via simulation sound.
This problem has not been addressed previously in~\cite{han2004reachability}.

Assume that the actual values of the bounds of $e_1$ and $e_2$ are $\bar{e}_{1}$ and $\bar{e}_{2}$ respectively, and the values of error bounds we get from simulation are $\tilde{e}_1$ and $\tilde{e}_2$. 
The numerical inaccuracy in simulation can be formulated as $\bar{e}_1 = \tilde{e}_1\pm \epsilon_1$ and $\bar{e}_2 = \tilde{e}_2\pm \epsilon_2$.
The actual overall bound $\delta$ is $\bar{e}_{1} + \bar{e}_{2}$ which satisfies the following constraint:
\begin{equation*}
\delta = \bar{e}_{1} + \bar{e}_{2} = \tilde{e}_1 + \tilde{e}_2 \pm \epsilon_1 \pm \epsilon_2.
\end{equation*}

From the above equation, it is easy to see that if we use the simulation bounds of $e_1$ and $e_2$ to calculate $\delta$, then the result may be unsound due to numerical issues (i.e. if $\pm \epsilon_1 \pm \epsilon_2 > 0$). 
To handle this, we can assume that the absolute numerical error in simulation $|\epsilon_i|,~i = 1,2$ is smaller than $\gamma$ percent of the simulation value $\tilde{e}_i,i=1,2$.
Then, the simulation error bound can be used as a sound result by bloating the simulation error bound using following equation.:

\begin{equation*}
\delta = (1+\gamma)(\tilde{e}_1 + \tilde{e}_2).
\end{equation*}
\commentTaylor{This is weak, fix the $\%$ thing finally. Use a decimal if necessary, but this percentage thing is weird as that is typically mod and it will be confusing.}

The soundness of error bounds $\delta$ computed using~\thmref{e1_bound},~\thmref{e1_bound_opt},~\thmref{e2_bound}, and the methods of~\cite{girard2007approximate} are guaranteed since these methods do not have numerical issues that may arise in simulation-based methods.   

\subsection{Output abstraction and $\delta$-approximation relation}

\begin{lemma}\lemlabel{output_abstraction}
Given an asymptotically stable LTI system $M_n$, there exists a $k$-dimensional output abstraction $M^{\delta}_{k}$ of $M_{n}$.
\end{lemma}
\begin{proof}
Assume that we have an asymptotically stable LTI system $M_n$, using balanced truncation method in~\secref{reduction}, we can obtain a $k$-dimensional reduced system $M_k$.
Moreover, from~\secref{error_bound}, for any $x(0)\in X_0(M_n)$, there exists $x_r(0) = SHx(0)\in X_0(M_k)$ such that the distance between each pair of output $(y_i(t),y_{r,i}(t))$ of the two systems is bounded by a finite positive real $\delta_i$, and this applies in every dimension $1 \leq i \leq p$~\eqref{delta}.
Hence, we can conclude that there exists a $k$-dimensional output abstraction $M_{k}^{\delta}(y_r|\{x_r,u\})$ $\langle A_r,B_r,C_r \rangle$ of $M_n$.
\end{proof}

There is a relationship between the output abstraction and $\delta$-approximate (bi)simulation relations~\cite{girard2007approximate} given as follows. 

Consider two dynamic systems:
\begin{align*}
\Sigma:~\dot {x}(t) &= f_1(x(t),u(t)), \\
y(t)& = g_1(x(t)), \\
\tilde \Sigma:~\dot {\tilde x}(t) &= f_2(\tilde{x}(t),u(t)), \\
\tilde y(t)& = g_2(\tilde{x}(t)).
\end{align*}

The central notion of approximate bisimulation is to characterize and quantify the distance between the outputs $y(t)$ and $\tilde y(t)$ generated by system $\Sigma$ and $\tilde \Sigma$ with the same input $u(t)$.
%

\begin{definition} \cite{girard2007approximate}\deflabel{bisim}
A relation $\mathscr {R}_\delta \subseteq \mathbb{R}^{n_x} \times \mathbb{R}^{\tilde n_x}$ is called a $\delta$-approximate bisimulation relation between systems $\Sigma$ and $\tilde \Sigma$, of \emph{precision} $\delta$, if, $\forall t \in \mathbb{R}_{\ge 0}$ and for all $(x(t),\tilde x(t)) \in \mathscr {R}_\delta$:
\begin{enumerate}
\item $\left\| {y(t)-\tilde y(t)} \right\| \le \delta$,
\item $\forall u(t) \in \mathcal{U}$, $\forall$ solutions $x(t)$ of $\Sigma$, $\exists$ a corresponding solution $\tilde x(t)$ of $\tilde \Sigma$ such that $(x(t),\tilde x(t)) \in \mathscr{R}_\delta$,
\item $\forall u(t) \in \mathcal{U}$, $\forall$ solutions $\tilde x(t)$ of $\tilde \Sigma$, $\exists$ a corresponding  $x(t)$ of $ \Sigma$ such that $(x(t),\tilde x(t)) \in \mathscr{R}_\delta$.
\end{enumerate}
If these conditions are met, we say systems $\Sigma$ and $\tilde \Sigma$ are approximately bisimilar with precision $\delta$, denoted by $\Sigma \sim_\delta \tilde \Sigma$.
\commentTaylor{You have to fix the ``satisfies'' thing. What does that mean?}
\end{definition}

Parameter $\delta$ measures the similarity of two systems $\Sigma$ and $\tilde \Sigma$.
In particular, $\mathscr {R}_0$ with $\delta=0$ recovers the exact bisimulation relation.
However, in most situations, the value of $\delta$ has to be greater than zero for two bisimilar systems, then the problem of calculating a tight estimate of $\delta$ is of the most importance in using approximate bisimulation relations for verification.

In this context, we relate the precision $\delta$ with the overall error bound $\delta = [\delta_1, \delta_2,\ldots,\delta_p]^T$ developed earlier~\eqref{delta}, to obtain the following proposition to establish an approximate bisimulation relation between the full-order system $M_n$ and its output abstraction $M_{k}^{\delta}$.

\begin{proposition}
For the full-order LTI system $M_{n}(y|\{x,u\})\langle A,B,C \rangle$ and the $k$-reduced order system $M_{k}(y_r|\{x_r,u\})\langle A_r,B_r,C_r \rangle$ by~\eqref{red_sys}, there exists an approximate bisimulation relation $\mathscr {R}_\rho$ such that $M_n \sim_\rho M_k$, where $\rho =\left\|\delta\right\|$.
\end{proposition}

\begin{proof}
For output $y(t)$ and $y_r(t)$ generated by $M_{n}$ and $M_{k}$, we have:
\begin{equation*}
\begin{split}
\left\|y(t)-y_r(t)\right\| =\sqrt{\sum\nolimits_{1}^{p}{(y^i(t)-y_{r}^{i}(t))^2}}.
\end{split}
\end{equation*}

Then, with the error bound $\delta_i$, $i=1,2\ldots,p$, computed by~\eqref{delta}, we have:
\begin{equation*}
\left\|y(t)-y_r(t)\right\| \leq \sqrt{\sum\nolimits_{1}^{p}{\delta_i^2}}=\left\|\delta\right\|=\rho.
\end{equation*}

According to~\defref{bisim}, and by either~\thmref{e1_bound} or~\thmref{e1_bound_opt} for the $e_1$ error bound and by~\thmref{e2_bound} for the $e_2$ error bound, the approximate bisimulation relation $\mathscr{R}_\rho$ with precision $\rho$ is established.
\commentTaylor{Refer to your earlier theorems where you establish that these error bounds are sound. Otherwise, why do you have these earlier theorems? Use your existing results.}
\commentTaylor{Also, you need to have a fixed constant bound for $e_1$ and $e_2$, right? This is why the time aspect is very confusing: you need the max over all time for these bounds for them to be sound. Make this clear earlier and now, and you should probably just have something here that is like: $\rho = e_1 + e_2$ or whatever makes sense.}

\end{proof}

\begin{remark}
It should be emphasized that there is a difference between the output abstraction and the $\delta$-approximate bisimulation relation since we compute the distance element to element between the outputs of two systems, i.e., $\left\|y_i - y_{r,i}\right\| \leq \delta_i$.
Our more precise result can produce a tighter transformed safe and unsafe specifications, which will be clarified in the next section.
\end{remark}

\subsection{Computational time complexity to compute the error bound}
Assume that the average time for one simulation is $\bar{t}$, then the time for the simulation-based approach~\cite{han2004reachability} to compute the error bound will be $\bar{t} \times N$, where $N$ is the total number of simulations. 
To analyze the time complexity of the simulation-based approach, we need to determine $N$.
For an $n$-dimension system with $m$ inputs and $p$ outputs, the number of vertices in polyhedral initial set is $2^n$.
Therefore, the number of simulations that need to be done to determine the bound of $e_1$ is $2^n$.
Similarly for $e_2$, as discussed in the previous section, the number of simulations for determining the bound of $e_2$ is $m$. 
Overall, the number of simulations $N$ need to be done in the worst case is $N = 2^n + m$, and the overall simulation time needed is $O(\bar{t} \times (2^n + m))$.

In~\cite{girard2007approximate}, to compute the error bound, this method solves two LMI and quadratic optimization problems on the sets of initial state and inputs. 
To estimate the time complexity of this method, we need to calculate the number of decision variables first. 
For the $n$-dimensions system, the number of decision variables is $(n^2+n)/2$. 
The number of LMI constraints related to this method is $2$.
Consequently, the time complexity for solving two LMI constraints using interior point algorithms can be estimated by $O([(n^2+n)/2]^{2.75}\times 2^{1.5})$~\cite{vandenberghe1994positive,nesterov1994interior}.
The time complexity for solving the optimization problem in~\cite{girard2007approximate} using interior point algorithms is $O([(n^2+n)/2]^3)$~\cite{vandenberghe1994positive,nesterov1994interior}. 
Totally, the time complexity in computing the error bound of the method proposed in~\cite{girard2007approximate} is $O([(n^2+n)/2]^3) + O([(n^2+n)/2]^{2.75}\times 2^{1.5})$, which can be bounded as $O(n^6)$.
\commentTaylor{A bit confused, are you sure this is $n^6$, it seems possibly high (but maybe not outlandish)? If you got these complexity estimates from a textbook, I guess it may be safe to use them.}

In our approach, it is easy to see that~\thmref{e1_bound} computation mainly relates to solving the optimization problem to find $\sup_{x_0 \in X_0} \left\| \bar{x}_{0} \right\|$.
This optimization problem can be done in two steps. 
The first step is to find the upper bound and lower bound of $\tilde{x}_0$ by solving the linear optimization problem defined by $min(max)Hx_0,~x_0 \in X_0$.
\commentTaylor{Refer back to your theorem...?}
Then, the supremum of the Euclidean-norm of $\bar{x}_0$ can be easily obtained.
If we use interior point algorithms for this problem, the time complexity of our approach using~\thmref{e1_bound} is $O((n+1)^{3.5})$~\cite{nesterov1994interior}. 
If we use the optimization method proposed in~\thmref{e1_bound_opt}, we first need to solve the eigenvalue problem (EVP) subject to two matrix inequalities that has time complexity $O([((n+k)^2+n+k)/2]^{2.75}\times 2^{1.5})$ if using interior point algorithms~\cite{vandenberghe1994positive}.  
Then, we need to solve the quadratic optimization problem that has time complexity $O((n+k)^3)$ if we use the interior point algorithm~\cite{ye1989extension}.  
Note that the computation cost of~\thmref{e2_bound} in our approach is small compared to~\thmref{e1_bound} and~\thmref{e1_bound_opt}. 

\tabref{time_complex} shows the simplified time complexity analysis of different approaches to compute the error bound $\delta$. 
As can be seen from the above discussion and~\tabref{time_complex}, in terms of time complexity, our approach using~\thmref{e1_bound} and~\thmref{e2_bound} is more efficient when dealing with high-dimensional systems while using~\thmref{e1_bound_opt} does not improve the time complexity.
The computation time of different methods are measured and discussed in detail in~\secref{casestudies}.  
\begin {table}
\centering
\small
  \begin{tabular}{ |c|c| }
    \hline
		\cite{girard2007approximate} & $O(n^6)$ \\ \hline
		\cite{han2004reachability}   & $O(2^n + m)$ \\ \hline
		~\thmref{e1_bound} & $O(n^{3.5})$ \\ \hline
    ~\thmref{e1_bound_opt} & $O((n+k)^{5.5})$ \\   
   \hline
  \end{tabular}
	\caption {Time complexity of different methods to compute the error bound $\delta$.}
	\tablabel{time_complex} 
\end{table}

\section{Safety Verification with Output Abstractions}
\seclabel{verificationusingoutputabstraction}
In this section, we focus on answering two critical questions: (a) how can an output abstraction be used to verify safety specifications of the original, full-order system? (b) How can an appropriate output abstraction be derived automatically?
To answer the first question, safety specifications must be transformed from those over the states of the full-order system to its output abstraction. 

For the general approximate bisimulation relation $\Sigma \sim_\delta \tilde \Sigma$, we usually use $\delta$-neighborhood to transform the safe and unsafe set.
\begin{proposition}
If $\Sigma \sim_\delta \tilde \Sigma$, then the following statements are true:
\begin{enumerate}
\item System $\Sigma$ is safe if $R_{[t_0,t_f]}(\tilde{\Sigma}) \cap \mathcal{N}(U(\Sigma),\delta) = \emptyset$ 
\item System $\Sigma$ is unsafe if $R_{[t_0,t_f]}(\tilde{\Sigma}) \cap \mathcal{N}(S(\Sigma),\delta) \neq \emptyset$ 
\end{enumerate}
where $\mathcal{N(\cdot,\delta)}$ denotes the $\delta$-neighborhood of a set.
\end{proposition}
\begin{remark}
$\delta$-neighborhood is a general approach to transform safety specification. However, in some cases when the safe and unsafe specifications are describes by polytopes or ellipsoids which are usually used in practical systems, for example in SpaceEx~\cite{frehse2011spaceex} and Ellipsoid Toolbox~\cite{kurzhanskiy2006ellipsoidal}, a more precise transformed safety specification can be derived by our element to element approach in~\secref{optainoutputabstraction} since $\delta_i \leq \left\|\delta\right\|$, thus it performs better than the $\delta$-neighborhood approach.
\end{remark}

In the following, we present detailed algorithms to transform the safety specifications described by convex polytopes and ellipsoids. 
 
\subsection{Transforming Safety Specifications}

\subsubsection{$S(M_n)$ as Convex Polytopes}
Assume that the safety specification of the full-order system is of the form:%
\begin{equation} %
\eqlabel{SP_full_order_polytopes}%
 S(M_n) = \{y \in \mathbb{R}^p |~\Gamma y + \Psi \leq 0 \},
\end{equation}
where $\Gamma = [\alpha_{ij}] \in \mathbb{R}^{q\times p}$ and $\Psi = [\beta_i] \in \mathbb{R}^{q}$.
\begin{lemma} \lemlabel{3}
Given $S(M_n)$ described by~\eqref{SP_full_order_polytopes}, then $S(M_{k}^{\delta})$ and $U(M_{k}^{\delta})$ defined as follows guarantee the safety relation~\eqref{safety_relation}.
\begin{equation}\eqlabel{SP_red_order_polytopes}
\begin{split}
&S(M_k^{\delta}) = \{y_r \in \mathbb{R}^p|~\Gamma y_r + \overline{\Psi} \leq 0\}, \\ 
&U(M_k^{\delta}) = \{y_r \in \mathbb{R}^p|~\Gamma y_r + \underline{\Psi} > 0\}, \\ 
&\overline{\Psi} = \Psi + \Delta,~\underline{\Psi} = \Psi - \Delta, \\
&\Delta = [\Delta_i] \in \mathbb{R}^{q},~\Delta_i = \sum_{j=1}^p|\alpha_{ij}|\delta_j. 
\end{split}
\end{equation}
\end{lemma}
The proof is given in~\appref{Appendix_alg}.

\subsubsection{$S(M_n)$ as Ellipsoids}
Assume that the safety specification of the full-order system is described by an ellipsoid with radius $R$ as:%
\begin{equation}\eqlabel{SP_full_order_ellipsoid}%
 S(M_n) = \{y \in \mathbb{R}^p|~ (y-a)^TQ(y-a) \leq R^2\},%
\end{equation}
where $a$ is the center of the ellipsoid and $Q \in \mathbb{R}^{p\times p}$ is a symmetric positive definite matrix. 

Since $Q$ is a symmetric matrix, there exists an orthogonal matrix $E = [l_1,l_2,..,l_p]$ $= [\gamma_{ij}] \in \mathbb{R}^{p\times p}$ such that $E^TQE = \Lambda = diag(\lambda_1,~\lambda_2,...,\lambda_p)$, where $\lambda_i$ ($>0$) is eigenvalue of $Q$ and $l_i$ is the eigenvector of $Q$ corresponding to $\lambda_i$.
The transformed safety and unsafe specifications of the output abstraction can be obtained with the following lemma. 

\begin{lemma} \lemlabel{4}
Given $S(M_n)$ described by~\eqref{SP_full_order_ellipsoid}, then $S(M_{k}^{\delta})$ and $U(M_{k}^{\delta})$ defined as follows guarantee the safety relation~\eqref{safety_relation}:%
\begin{align*}
\begin{split}
 &S(M_k^{\delta}) = \{y_r \in \mathbb{R}^p |~ (y_r-a)^TQ(y_r-a) \leq (R-\Delta_R)^2\},\\
 &U(M_k^{\delta}) = \{y_r \in \mathbb{R}^p |~ (y_r-a)^TQ(y_r-a) > (R+\Delta_R)^2\},\\
 &\Delta_R = \sqrt{\sum_{i=1}^p[\lambda_i(\sum_{j=1}^p|\gamma_{ij}|\delta_j)^2]}.
\end{split}
\end{align*}
 
\end{lemma}

The proof of this result is given in~\appref{Appendix_alg}.

We can see that the transformed safety specification of the output abstraction is also an ellipsoid (with smaller radius $R-\Delta_R$) located inside the original ellipse defining the safety specification of the full-order system.
Meanwhile the corresponding transformed unsafe specification is defined by the region outside the larger ellipse with the radius $R+\Delta_R$. 

\subsubsection{Further discussion}

Lemmas~\ref{lem:3} and~\ref{lem:4} can be applied directly for the system with bounded safety specification (i.e., the safety specification is bounded and thus, the unsafe specification is unbounded).
Conversely, it is easy to see that our method can also be applied for the system with unbounded safety specification (i.e., the unsafe specification is bounded).
Assume that the unsafe specification of the full-order system is bounded by a convex polytope as follows: 
\begin{equation} 
\eqlabel{UP_full_order_polytopes}
 U(M_n) = \{y \in \mathbb{R}^p |~\Gamma y + \Psi \leq 0 \},
\end{equation}
where $\Gamma$ and $\Psi$ are defined as in~\eqref{SP_full_order_polytopes}.
Then, the transformed unsafe specification for the output abstraction is defined by:
\begin{equation}\eqlabel{UP_red_order_polytopes}
U(M_k^{\delta}) = \{y_r \in \mathbb{R}^p|~\Gamma y_r + \underline{\Psi} \leq 0\}, 
\end{equation} 
where $\underline{\Psi}$, $\Delta$ are defined as in~\lemref{3}. 

Similarly, suppose the unsafe specification of the full-order system is bounded by an ellipsoid with radius $R$: 
\begin{equation}\eqlabel{UP_full_order_ellipsoid}
 U(M_n) = \{y \in \mathbb{R}^p|~ (y-a)^TQ(y-a) \leq R^2\},
\end{equation}
where $a$ and $Q$ are defined as in \eqref{SP_full_order_ellipsoid}.
Then, the corresponding transformed unsafe specification for the output abstraction is bounded by the following:

\begin{figure}[t]%
	\centering%
		\centering%
		\includegraphics[width=0.67\columnwidth]{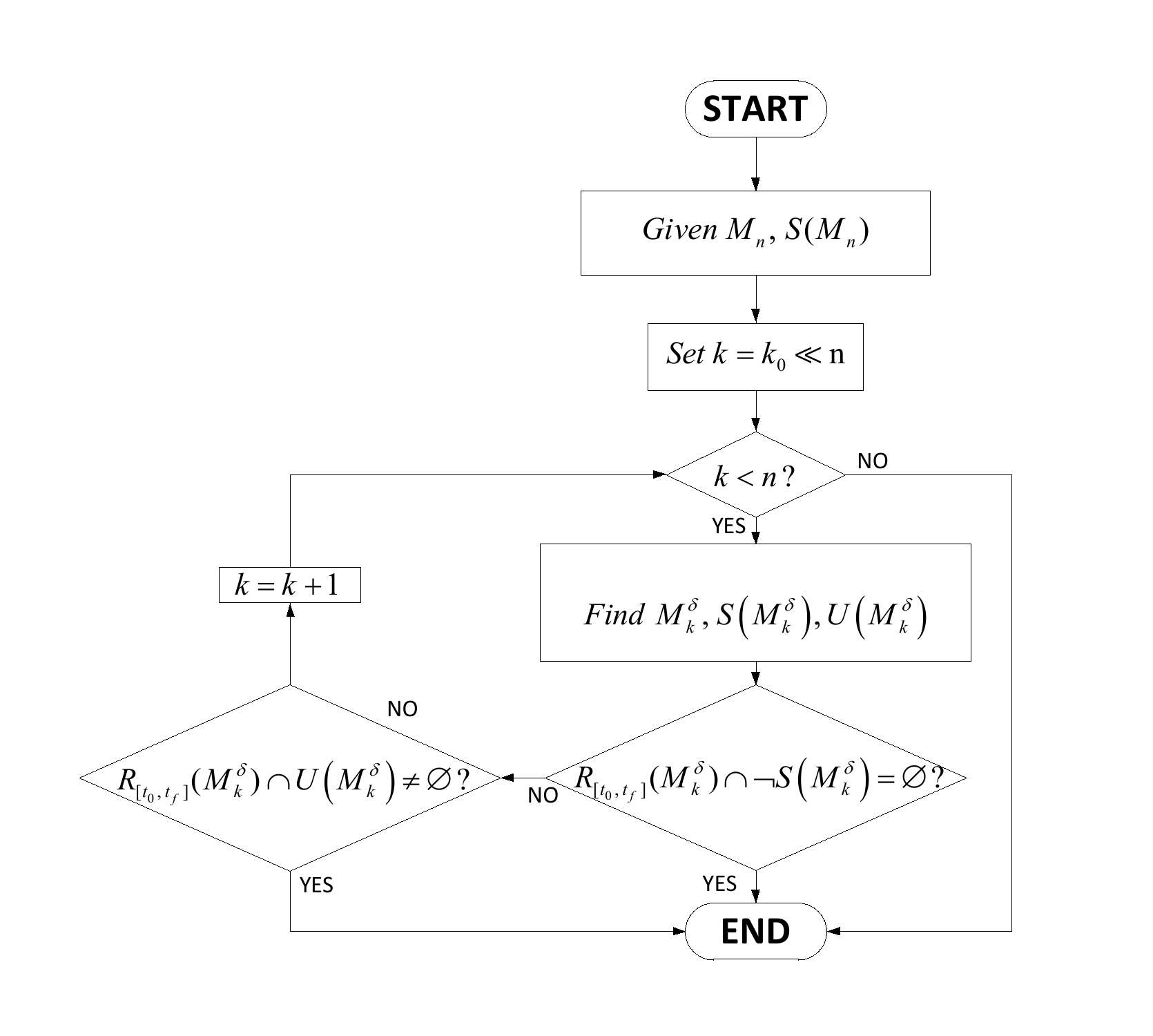}%
		\caption{Semi-algorithm for automatic safety verification using output abstraction.}%
		\figlabel{algorithm}%
		\vspace{-1.5em}%
\end{figure}

\begin{equation}\eqlabel{UP_red_order_ellipsoid}
 U(M_k^{\delta}) = \{y_r \in \mathbb{R}^p |~ (y_r-a)^TQ(y_r-a) \leq (R+\Delta_R)^2\},
\end{equation}
where $\Delta_R$ is the same as in~\lemref{4}.

The proof for the above transformation is the same as in Lemmas~\ref{lem:3} and~\ref{lem:4}.
Geometrically, the transformed unsafe specification of the output abstraction is a larger region containing the unsafe region of the full-order system.  

\subsection{Safety Verification with Output Abstraction}


%
\begin{lemma}\lemlabel{soundness}
Given a asymptotic stable LTI system $M_n$, whenever its output abstraction $M_{k}^{\delta}$ is safe or unsafe, it is sound to claim that the system $M_n$ is safe or unsafe respectively.   
\end{lemma}
\begin{proof}
According to~\lemref{output_abstraction}, for a stable LTI system, there exists an output abstraction $M_{k}^{\delta}$. From the definition of the output abstraction, we can see that, for any output trajectory of $M_n$, there exists a corresponding output trajectory of $M_{k}^{\delta}$ such that the distance between two trajectories is always bounded by a sound $\left\|\delta\right\|$. Moreover from Lemmas~\ref{lem:3} and~\ref{lem:4}, because the transformed specifications ($S(M_{k}^{\delta})$ and  $U(M_{k}^{\delta})$) satisfy the safety relation~\eqref{safety_relation}, thus when each output trajectory of $M_{k}^{\delta}$ satisfies the transformed safety specification $S(M_{k}^{\delta})$, its corresponding output trajectory of $M_n$ also satisfies the original safety specification $S(M_n)$. That means, if all output trajectories of $M_{k}^{\delta}$ satisfy the transformed safety specification $S(M_{k}^{\delta})$, then all output trajectories of $M_n$ also satisfy the original safety specification $S(M_n)$. Consequently, when the output abstraction is safe, it is sound to claim that the full-order system is safe. A similar proof can be given for the unsafe case. The proof is completed. 
\end{proof}

So far, the process of obtaining an output abstraction and its safety specifications from a given high-dimensional linear system and safety requirements can be done automatically.
A semi-algorithm for automatic safety verification of a high-dimensional system using its output abstraction is depicted in~\figref{algorithm}.
The method finds a $k$-dimension output abstraction ($k = k_0$ initially, where $k_0$ is given by the user) and the corresponding safety specification, then checks the safety of the output abstraction.
The method may not terminate as checking the safety of the output abstraction is undecidable since it involves computing the reachable states for a linear system, so it is a semi-algorithm.
However, in practice, tools such as SpaceEx may terminate for time-bounded overapproximate reachability computations for systems of small enough dimensionality.\footnote{Note that we cannot use SpaceEx to conclude a system is unsafe because it computes over-approximations of the actual set of reachable states. We could conclude unsafety if under-approximations were available, but not many tools compute under-approximations.}
If it is safe (or unsafe) then the method stops with the conclusion that the full-order system is safe (or unsafe) and returns the current $k$-order abstraction.
If the safety of the output abstraction cannot be verified (i.e., it is indeterminate), then the algorithm will repeat the same process for another output abstraction whom the order is increased by $1$ from the order of the current abstraction.

\begin{figure*}[t]%
    \centering%
    \begin{subfigure}[t]{0.45\textwidth}%
				\centering%
				\includegraphics[width=\textwidth]{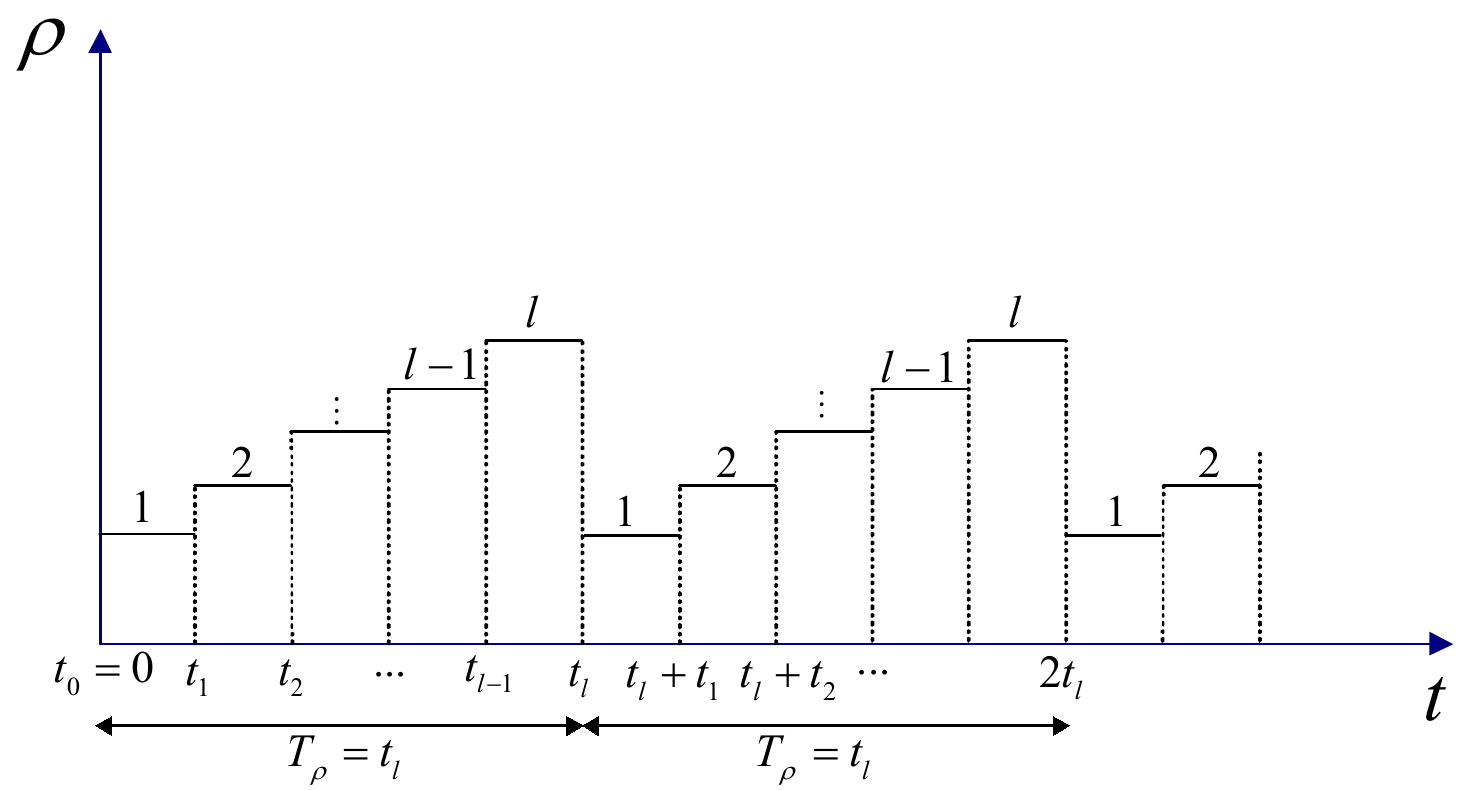}%
				\caption{Periodic switching function for PSS.}
				\figlabel{Switching_function}%
    \end{subfigure}%
~
    \begin{subfigure}[t]{0.45\textwidth}%
				\centering%
				\includegraphics[width=\textwidth]{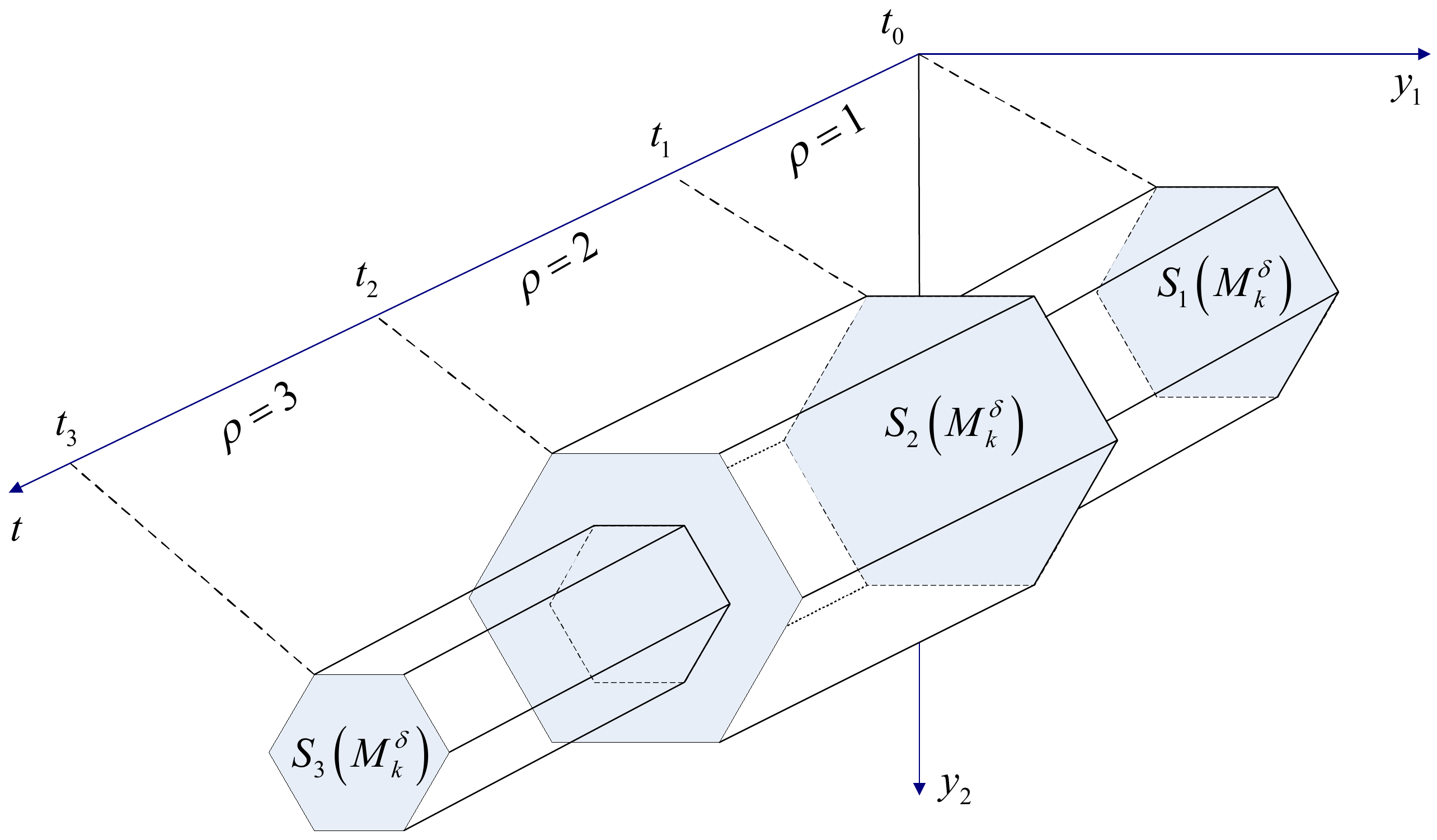}
				\caption{Transformed safety specification of a PSS's output abstraction.}
				\figlabel{PSS_SP_transformation}
    \end{subfigure}%
    \caption{}%
		\figlabel{iss}%
		\vspace{-1.25em}%
\end{figure*}%

We have discussed how to use an output abstraction for verification and proposed an semi-algorithm to obtain automatically an appropriate output abstraction.
In the next section, the method is extended to a class of periodically switched linear systems.

\section{Output Abstraction and Verification for Periodically Switched Systems}
\seclabel{abstraction_hybrid}

In this section, we extend our previous results to verify the safety of periodically switched systems (PSSs). 

\subsection{Output Abstraction for Periodically Switched Systems}

To define the class of PSSs, we need to define the notion of periodical switching signal $\rho(t)$. 

Let $\mathcal{P}$ $=$ $\{1, 2, ..., l-1, l\}$, where $l \in \mathbb{N}$ and $\mathcal{T} = \{t_0, t_1, t_2, ..., t_{l-1}, t_{l}\}$.
The periodical piecewise constant function $\rho(t)$ describing the switching signal is assumed to be right-continuous everywhere and is defined by:%
\[\left\{ \begin{array}{l}
\rho \left( t \right) = i,\;\;{t_{i - 1}} \le t < {t_i},\;\;i \in \mathcal{P}\\
\rho \left( {k \times {t_l} + t} \right) = \rho \left( t \right),\;{t_0} \le t < {t_l},\;\;k = 0,\;1,\;2,\;...\;
\end{array} \right.\]

\figref{Switching_function} illustrates a periodic switching function with the assumption that $t_0 = 0$.
For simplicity, we omit the time argument of switching function in the rest of the paper.
The class of PSS is given by:
\begin{equation}
\eqlabel{PSS}
\begin{split}
    &\dot{x}(t) = A_{\rho}x(t) + B_{\rho}u(t),~~ t \notin \mathcal{T}, \\
    &x(t^+) = r_{\rho}(x(t)),~~ t \in \mathcal{T},~~ x(t^+) \in X^o_{\rho}   \\   
    &y(t) = C_{\rho}x(t),
\end{split}
\end{equation}
where $x(t) \in \mathbb{R}^{n}$ is the system state, $y(t) \in \mathbb{R}^{p}$ is the system output, $u(t) \in \mathbb{R}^{m}$ is the control input, and $(A_{\rho}, B_{\rho}, C_{\rho})$ are system matrices in the ``mode'' $\rho$; $r_{\rho}(\cdot)$, and $X^o_{\rho}$ respectively are the bounded resetting function and the initial set of state in the ``mode'' $\rho$.
In particular, the first equation of~\eqref{PSS} describes the continuous dynamics of the PSS and the second equation represents the resetting law.
For brevity, we denote $M_n$ as the full-order PSS defined by~\eqref{PSS}.

The above class of switched systems with resetting functions can be found in some supervisory control structures with the employment of impulsive control techniques.
Once the system state is observed to reach some unsafe large values far away from initial set, some impulsive control schemes can be activated to abruptly drag the state back to the initial set to ensure the safety of the system.
Note that in this paper, the resetting functions are activated periodically along with time (i.e., they do not depend on state variables).

Since the initial state for each mode is specified in each switching time, using the same approach as for linear continuous system, we can find the reduced-order representation called the local output abstraction for each mode of $M_n$ that produces the output signal capturing the full-order system output signal within an computable error bound.
In other words, we can derive the output abstraction $M_{k}^{\delta}$ for the full-order PSS $M_n$ as:
\begin{equation}\eqlabel{PSS_red}
\begin{split}
    &\dot{x}_r(t) = \breve{A}_{\rho}x_r(t) + \breve{B}_{\rho}u(t),~~ t \notin \mathcal{T}, \\
    &x_r(t^+) = ST_{\rho}x(t^+),~~ t \in \mathcal{T}   \\   
    &y_r(t) = \breve{C}_{\rho}x_r(t),
\end{split}
\end{equation} 
where $x_r(t) \in \mathbb{R}^{k}$, $y_r(t) \in \mathbb{R}^{p}$ and the system matrices $(\breve{A}_{\rho}, \breve{B}_{\rho}, \breve{C}_{\rho})$ of the abstraction in mode $\rho$ can be obtained using the balanced truncation reduction method~\secref{optainoutputabstraction} defined by the balancing transformation matrix $T_{\rho}$ and the cutting matrix $S = \left(
\begin{array}{cc}
    I_{k \times k} & 0_{k \times (n-k)} \\
  \end{array}
\right)$. 


The output signal of the abstraction $M_{k}^{\delta}$ captures the output of the full-order system $M_{n}$ at all the times $t$ within a piecewise linear constant error bound $\delta_{\rho} = [\delta^{\rho}_1, \delta^{\rho}_2, ..., \delta^{\rho}_p]^T$ which can be computed in the same manner as in the case of linear continuous systems in~\secref{optainoutputabstraction}. 

This summarizes the procedure for deriving the output abstraction for the full-order PSS.
Next, we consider how to use the output abstraction to verify safety of the full-order PSS.

\subsection{Verification for Periodically Switched Systems}

Similar to the case of linear continuous system, the key step for the verification process is that from the computed error bound $\delta_{\rho}$ and the given safety specification of the full-order PSS system, we determine the safety specifications for  the corresponding output abstraction that guarantees the safety relation~\eqref{safety_relation}.

Let us consider the first case that the safety specification of the full-order PSS~\eqref{PSS} is described by~\eqref{SP_full_order_polytopes} (i.e. as a polytope).
The transformed safety specifications for the corresponding output abstraction~\eqref{PSS_red} that satisfies the safety relation~\eqref{safety_relation} is defined by:
\vspace{-0.75em}%
\begin{align*}
\begin{split}%
&S(M_k^{\delta}) = \bigcup_{\rho} S_{\rho}(M_k^{\delta}),~~U(M_k^{\delta}) = \bigcup_{\rho} U_{\rho}(M_k^{\delta}), \\
&S_{\rho}(M_k^{\delta}) = \{y_r \in \mathbb{R}^p|~\Gamma y_r + \overline{\Psi}^{\rho} \leq 0\}, \\ 
&U_{\rho}(M_k^{\delta}) = \{y_r \in \mathbb{R}^p|~\Gamma y_r + \underline{\Psi}^{\rho} > 0\}, \\ 
&\overline{\Psi}^{\rho} = \Psi + \Delta_{\rho},~\underline{\Psi}^{\rho} = \Psi - \Delta_{\rho}, \\
&\Delta_{\rho} = [\Delta_i^{\rho}] \in \mathbb{R}^{q},~\Delta_i^{\rho} = \sum_{j=1}^p|\alpha_{ij}|\delta_j^{\rho}. 
\end{split}%
\end{align*}
\vspace{-1em}

Similarly, when the safety specification of the full-order PSS has the form of an ellipsoid as defined in~\eqref{SP_full_order_ellipsoid}, we can derive the corresponding transformed safety specification for the output abstraction as follows.

\begin{align*}
\begin{split}
 &S(M_k^{\delta}) = \bigcup_{\rho} S_{\rho}(M_k^{\delta}),~~U(M_k^{\delta}) = \bigcup_{\rho} U_{\rho}(M_k^{\delta}), \\
 &S_{\rho}(M_k^{\delta}) = \{y_r \in \mathbb{R}^p |~ (y_r-a)^TQ(y_r-a) \leq (R-\Delta_R^{\rho})^2\},\\
 &U_{\rho}(M_k^{\delta}) = \{y_r \in \mathbb{R}^p |~ (y_r-a)^TQ(y_r-a) > (R+\Delta_R^{\rho})^2\},\\
 &\Delta_R^{\rho} = \sqrt{\sum_{i=1}^p[\lambda_i(\sum_{j=1}^p|\gamma_{ij}|\delta_j^{\rho})^2]}.
\end{split}
\end{align*}

We have obtained the transformed safety specification for the output abstraction of a PSS.
Since the error bound $\delta_{\rho}$ varies along with time (i.e. it depends on what mode is being activated), the transformed safety specification for the output abstraction also varies along with time.
In addition, it is periodic because of the periodicity of the switching function $\rho(t)$.
\figref{PSS_SP_transformation} presents an example of the transformed safety specification of a PSS with two outputs.
From the transformed safety specification, we can verify straightforwardly the safeness of the output abstraction to conclude about safety of its full-order PSS.

We have considered how to obtain an output abstraction and use it to verify safety of PSS.
In the next section, several case of studies are presented to evaluate the benefits of our method.

\vspace{-2em}
\section{Case Studies and Evaluation}
\seclabel{casestudies}

\begin{figure*}[t]%
    \centering%
		\begin{subfigure}[t]{0.48\textwidth}%
				\centering%
				\includegraphics[width=\textwidth]{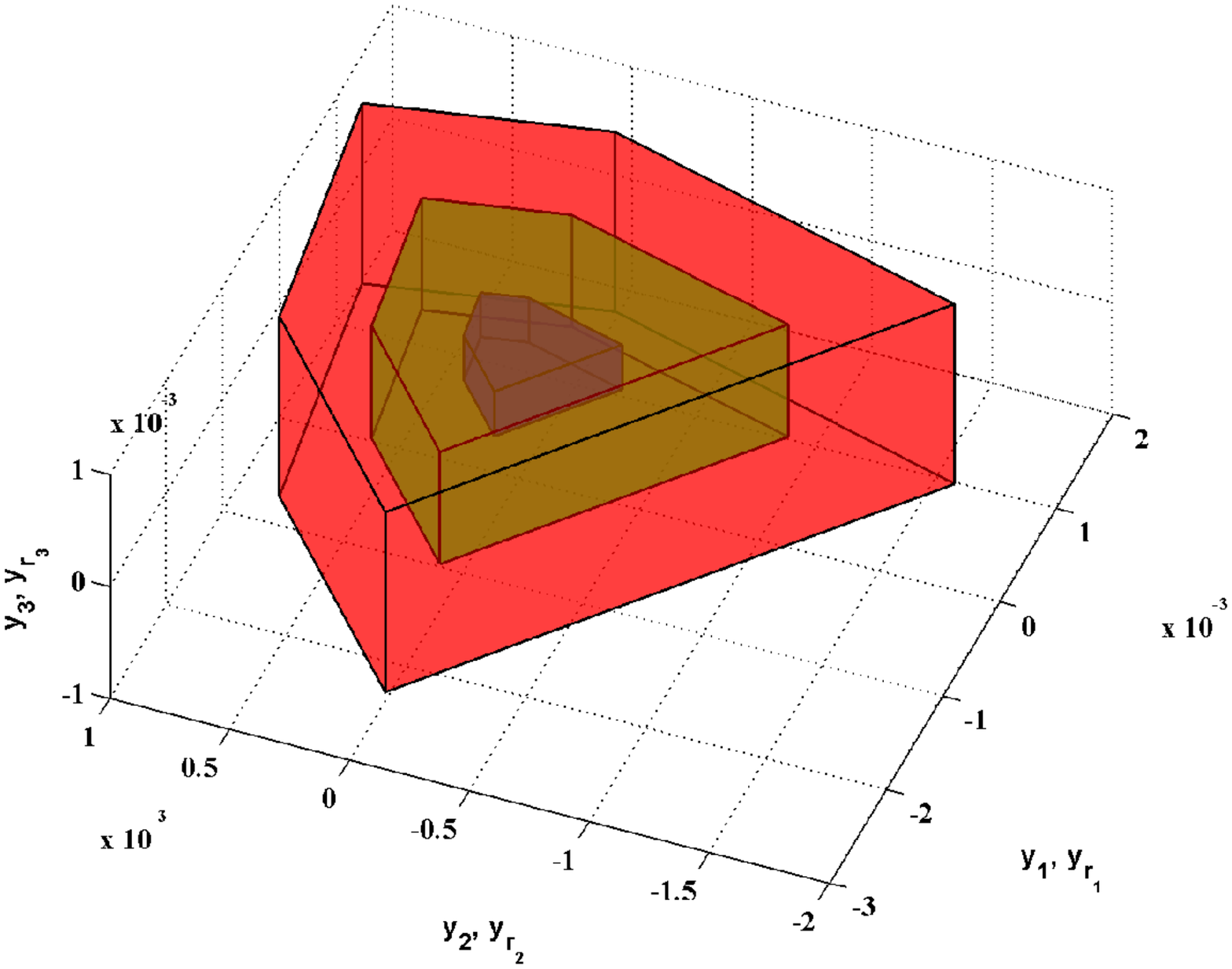}%
				\caption{Safety specification $S_{270} =$ of the full order ISS system which is the region inside the middle green polytopes and the transformed safety specifications of its 10-order output abstraction in which the safe region $S_{10}^{\delta}$ is inside the smallest blue polytopes and the unsafe region $U_{10}^{\delta}$ is outside the largest red polytopes.}%
				\figlabel{iss_SP}%
		\end{subfigure}%
~
		\begin{subfigure}[t]{0.4\textwidth}%
				\centering%
				\includegraphics[width=\textwidth]{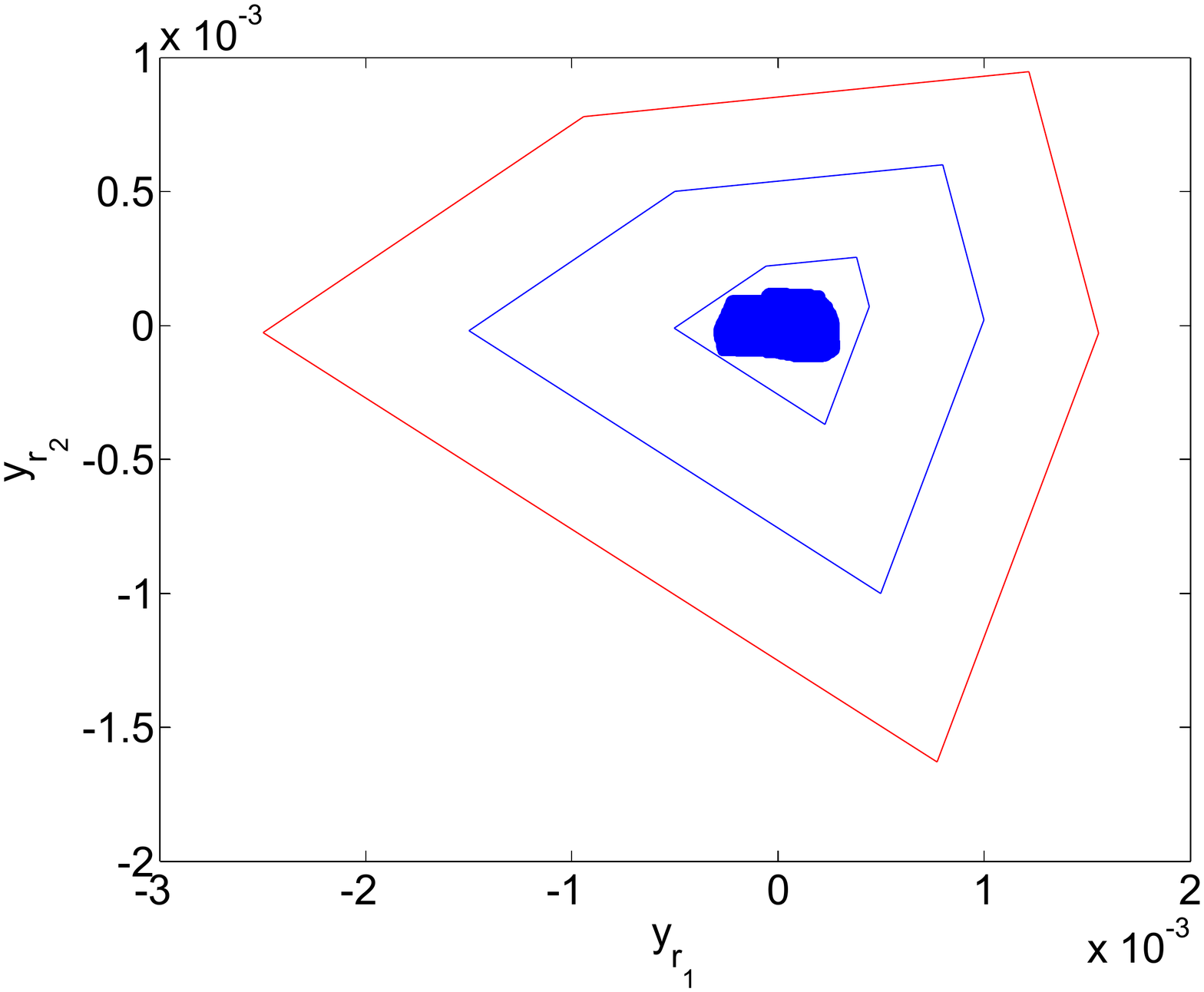}%
				\caption{Reachable set of $(y_{r_1}, y_{r_2})$ of the ISS's 10-order abstraction $M^{\delta}_{10}$ in the period of time $[0,~20s]$ and its safe region $S_{12}^{\delta}$ (inside the smallest blue polygon) and unsafe region $U_{12}^{\delta}$ (outside the red polygon).}%
				\label{verify_S_12}
		\end{subfigure}%

    \caption{}%
		\vspace{-1.25em}%
\end{figure*}%

\begin{figure*}[t]%
    \centering%
    \begin{subfigure}[t]{0.4\textwidth}%
				\centering%
				\includegraphics[width=\textwidth]{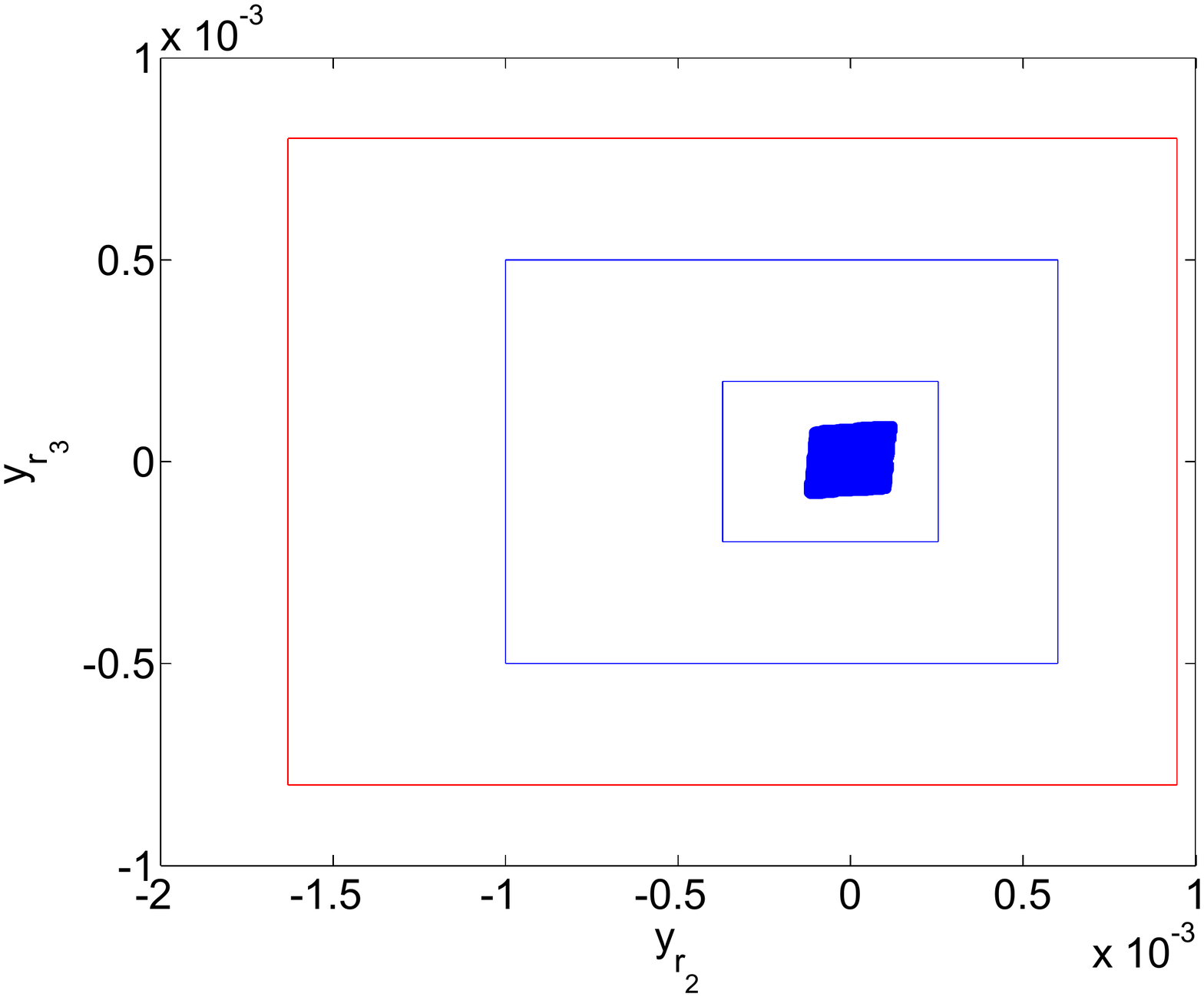}
				\caption{Reachable set of $(y_{r_2}, y_{r_3})$ of the ISS's 10-order abstraction $M^{\delta}_{10}$ in the period of time $[0,~20s]$ and its safe region  $S_{23}^{\delta}$ (inside the smallest blue polygon) and unsafe region $U_{23}^{\delta}$ (outside the red polygon).}
				\label{verify_S_23}
    \end{subfigure}%
~
    \begin{subfigure}[t]{0.4\textwidth}%
				\centering%
				\includegraphics[width=\textwidth]{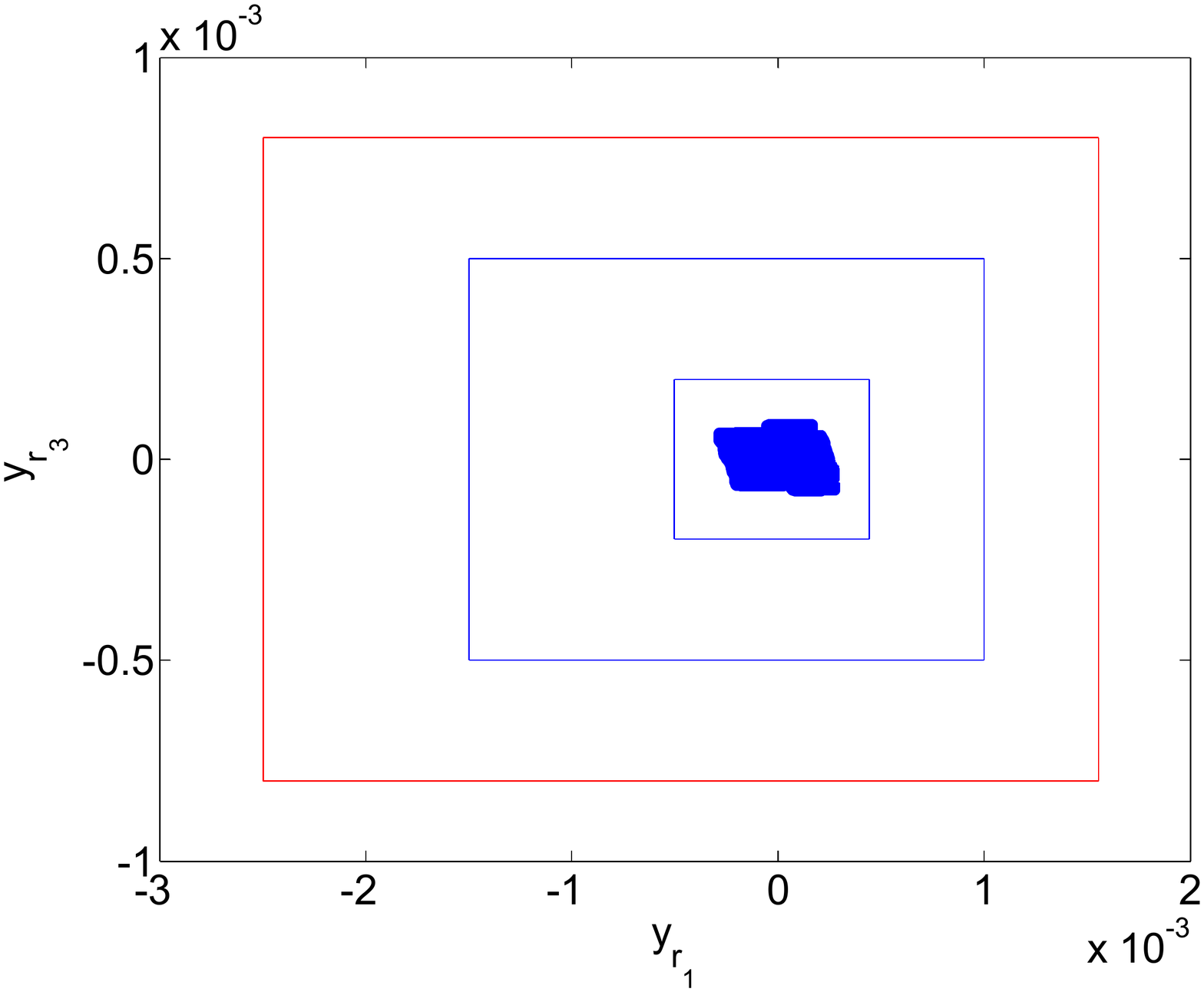}
				\caption{Reachable set of $(y_{r_1}, y_{r_3})$ of the ISS's 10-order abstraction $M^{\delta}_{10}$ in the period of time $[0,~20s]$ and its safe region $S_{13}^{\delta}$ (inside the smallest blue polygon) and unsafe region $U_{13}^{\delta}$ (outside the red polygon).}
				\label{verify_S_13}
    \end{subfigure}%

    \caption{}%
		\vspace{-1.25em}%
\end{figure*}%

\begin{table}[b] 
\resizebox{\textwidth}{!}{
  \begin{tabular}{|l|l|l|l|l|l|l|}
    \hline
    No. & Benchmark & Type & n & m & p \\
    \hline
    1 & Motor control system (MCS) & LTI & 8 & 2 &2 \\
		2 & Helicopter~\cite{frehse2011spaceex} & LTI & 28 & 6 & 2 \\
		3 & Building model (BM)~\cite{chahlaoui2002collection}& LTI & 48 & 1 & 1 \\
		4 & International space station (ISS)~\cite{chahlaoui2002collection} & LTI& 270 & 3 & 3 \\
		5 & Partial differential equation (Pde)~\cite{chahlaoui2002collection} & LTI & 84 & 1 &1 \\ 
		6 & FOM~\cite{chahlaoui2002collection} & LTI & 1006 & 1&1 \\ 
		7 & Synchronous position control system with 2 motors (SMS2) & PSS & 8 & 2&2 \\
    \hline
  \end{tabular}
	}
	\caption{Benchmarks for the order-reduction abstraction method in which $n$ is dimension of the system; $m$ and $p$ are the number of inputs and outputs respectively.}
	\tablabel{benchmark}
\end{table}

\begin{table*}[b]
\resizebox{\textwidth}{!}{
  \begin{tabular}{|p{2cm}|p{6.5cm}|p{2.5cm}|p{3.5cm}|}
    \hline
    \multirow{2}{*}{\textbf{Benchmark}} & \textbf{Initial set of states}  & \textbf{Input constraint} & \textbf{Safety specification} \\
		& $X_0 = \{x_0 \in \mathbb{R}^{n}|~lb(i) \leq x_0(i) \leq ub(i),~1 \leq i \leq n \}$ & $u = [u_1,\cdots,u_m]^T$ & $y = [y_1,\cdots,y_p]^T$ \\
    \hline

\multirow{3}{*}{\parbox{2cm}{Motor control system}}  &$lb(i) = ub(i) = 0,~ i=2,3,4,6,7,8,$ & $u_1 \in [0.16,~0.3],$ & unsafe region: \\                           
& $lb(2)=0.002,~ub(2)=0.0025,$ & $u_2 \in [0.2,~0.4].$ & $0.35 \leq y_1 \leq 0.4$, \\   
& $lb(3)=0.001,~ub(3) = 0.0015.$ &  & $0.45 \leq y_2 \leq 0.6$.\\  
		\hline

 \multirow{3}{*}{Helicopter} & $lb(i) = ub(i) = 0.1,~ i=1,4,5,6,7,$ & $u_i \in [-1,~1],$ & unsafe region: \\                          
	& $lb(2)=lb(3)=0.098,ub(2)=0.11,ub(3)=0.102,$ & $1 \leq i \leq 6.$ & $-1 \leq y_1 \leq 1$, \\
	& $lb(i) = ub(i) = 0,~8\leq i \leq 28.$ &  & $10 \leq y_2 $\\         
		\hline
		
 \multirow{3}{*}{\parbox{2cm}{Building model}} & $lb(i) = 0.0002,~ub(i)=0.00025,~1 \leq i \leq 10,$ &\multirow{3}{*}{$u_1 \in [0.8,~1].$} & unsafe region: \\                          
	& $lb(25)= -0.0001,~ub(25)=0.0001,$ & & $0.008 \leq y_1$ \\
	& $lb(i) = ub(i) = 0,~11 \leq i \leq 48,~i\neq 25.$ & &\\         
		\hline

 \multirow{3}{*}{\parbox{2cm}{Partial differential equation}} & $lb(i) = 0,~ub(i)=0,~1 \leq i \leq 64$ &\multirow{3}{*}{$u_1 \in [0.5,~1].$} & safe region: \\                          
	& $lb(i) = 0.001,~ub(i)=0.0015,~64 \leq i \leq 80,$ &  & $y_1 \leq 12$\\
	&  $lb(i) = -0.002,~ub(i)=-0.0015,~81 \leq i \leq 84.$ &  & \\         
		\hline
		
 \multirow{7}{*}{\parbox{2cm}{International space station}} &  \multirow{7}{*} {$lb(i) = -0.0001,~ub(i)=0.0001,~1 \leq i \leq 270.$}  & & Safe region:\\                          
	& & & $-461y_1 + 887y_2 + 0.67 \leq 0$, \\
	& & $u_1 \in [0,~0.1],$ & $-440y_1 - 898y_2 - 0.68 \leq 0$, \\ 
	& & $u_2 \in [0.8,~1],$ & $-76.7y_1 +997y_2 - 0.54 \leq 0$, \\  
	&	&	$u_3 \in [0.9,~1].$	& $898y_1 -440y_2 - 0.89 \leq 0$, \\
	&	&		& $945y_1 + 326y_2 - 0.95 \leq 0$, \\
	&	&		& $-0.0005 \leq y_3 \leq 0.0005$. \\
		\hline
		
 \multirow{3}{*}{\parbox{2cm}{FOM}} & $lb(i) = -0.0001,~ub(i)= 0.0001,~1 \leq i \leq 400$ &\multirow{3}{*}{$u_1 \in [-1,~1].$} & safe region: \\  & $lb(i) = 0.0002,~ub(i)=0.00025,~401 \leq i \leq 800,$ &  & $y_1 \leq 45$ \\
	&  $lb(i) = 0,~ub(i)= 0,~801 \leq i \leq 1006.$ &  & \\ 
	\hline    
  \end{tabular}
}	
  \caption{Initial states, input constraints and safety specification of LTI benchmarks.}
	 \tablabel{constraint}
\end{table*}

\begin{table*}[t] 
\resizebox{\textwidth}{!}{
\begin{tabular}{|p{1.5cm}|c|c|c|c|c|c|c|c|c|c|c|c|c|c|} 
\hline 
\multirow{2}{*}{\textbf{Benchmark}} &\multirow{2}{*}{\textbf{k}} & \multicolumn{2}{c|}{\textbf{\cite{girard2007approximate}}} & \multicolumn{3}{c|}{\textbf{\cite{han2004reachability}}} & \multicolumn{4}{c|}{\textbf{Mixed bound}} & \multicolumn{4}{c|}{\textbf{Theoretical bound}} \\ 
\cline{3-15} 
&  & {$\delta$} & $t(s)$ & $\delta$ & $t(s)$ & $N$ & $ e_1 $ & $ e_2 $ & $\delta$ &$t(s)$ & $e_1 $ & $e_2$ & $\delta$ &$t(s)$ \\ 
\hline 
\multirow{2}{*}{\parbox{1.5cm}{Motor control system}} 
& $5$ 
& $2.1$ & $0.93$ 
&$\begin{pmatrix}0.0021  \\ 0.047\end{pmatrix}$ & $0.17 $ & {$2^2+2$} 
&$\begin{pmatrix}0.00049  \\ 0.00062 \end{pmatrix}$ &$\begin{pmatrix}0.002  \\ 0.047 \end{pmatrix}$ &$\begin{pmatrix}0.0025  \\ 0.047 \end{pmatrix}$ & $0.92$ 
&$\begin{pmatrix}0.0098  \\ 0.0093 \end{pmatrix}$ &$\begin{pmatrix}0.53  \\ 0.53 \end{pmatrix}$ &$\begin{pmatrix}0.54  \\ 0.54 \end{pmatrix}$ & $0.15$  \\ 
\cline{2-15} 
& $4$ 
& $1.5$ & $0.64$ 
&$\begin{pmatrix}0.036  \\ 0.047\end{pmatrix}$ & $0.19 $ & {$2^2+2$} 
&$\begin{pmatrix}0.00086  \\ 0.00062 \end{pmatrix}$ &$\begin{pmatrix}0.035  \\ 0.047 \end{pmatrix}$ &$\begin{pmatrix}0.036  \\ 0.047 \end{pmatrix}$ & $0.9$ 
&$\begin{pmatrix}0.009  \\ 0.009 \end{pmatrix}$ &$\begin{pmatrix}0.91  \\ 0.91 \end{pmatrix}$ &$\begin{pmatrix}0.92  \\ 0.92 \end{pmatrix}$ & $0.15$  \\ 
\hline 

\multirow{3}{*}{\parbox{1.5cm}{Helicopter}} 
& $20$ 
& $0.84$ & $17$ 
&$\begin{pmatrix}7.1e-05  \\ 4.3e-05\end{pmatrix}$ & $0.6 $ & {$2^4+6$} 
&$\begin{pmatrix}0.0072  \\ 0.018 \end{pmatrix}$ &$\begin{pmatrix}5.7e-05  \\ 3.2e-05 \end{pmatrix}$ &$\begin{pmatrix}0.0073  \\ 0.018 \end{pmatrix}$ & $35$ 
&$\begin{pmatrix}0.28  \\ 0.95 \end{pmatrix}$ &$\begin{pmatrix}0.0017  \\ 0.0017 \end{pmatrix}$ &$\begin{pmatrix}0.28  \\ 0.95 \end{pmatrix}$ & $0.49$  \\ 
\cline{2-15} 
& $16$ 
& $28$ & $12$ 
&$\begin{pmatrix}0.00075  \\ 0.0013\end{pmatrix}$ & $0.56 $ & {$2^4+6$} 
&$\begin{pmatrix}0.0072  \\ 0.018 \end{pmatrix}$ &$\begin{pmatrix}0.0007  \\ 0.00087 \end{pmatrix}$ &$\begin{pmatrix}0.0079  \\ 0.019 \end{pmatrix}$ & $23$ 
&$\begin{pmatrix}0.28  \\ 0.95 \end{pmatrix}$ &$\begin{pmatrix}0.029  \\ 0.029 \end{pmatrix}$ &$\begin{pmatrix}0.3  \\ 0.97 \end{pmatrix}$ & $0.45$  \\ 
\cline{2-15} 
& $10$ 
& $160$ & $8.1$ 
&$\begin{pmatrix}0.024  \\ 0.038\end{pmatrix}$ & $0.55 $ & {$2^4+6$}
&$\begin{pmatrix}0.0085  \\ 0.021 \end{pmatrix}$ &$\begin{pmatrix}0.021  \\ 0.031 \end{pmatrix}$ &$\begin{pmatrix}0.03  \\ 0.053 \end{pmatrix}$ & $13$ 
&$\begin{pmatrix}0.27  \\ 0.93 \end{pmatrix}$ &$\begin{pmatrix}1  \\ 1 \end{pmatrix}$ &$\begin{pmatrix}1.3  \\ 1.9 \end{pmatrix}$ & $0.45$  \\ 
\hline 

\multirow{3}{*}{\parbox{1.5cm}{Building model}} 
& $25$ 
& $0.0096$ & $180$ 
&$0.0051$& $22$ &{$2^{11}+1$} 
&$0.013$ &$6.2e-05$ &$0.013$ &$130$ 
&$0.083$ &$0.0072$ &$0.09$ &$1$ \\ 
\cline{2-15} 
& $15$ 
& $0.069$ & $120$ 
&$0.005$& $18$ &{$2^{11}+1$} 
&$0.012$ &$0.00044$ &$0.013$ &$58$ 
&$0.078$ &$0.084$ &$0.16$ &$0.97$ \\ 
\cline{2-15} 
& $6$ 
& $0.1$ & $44$ 
&$0.0058$& $14$ &{$2^{11}+1$}  
&$0.011$ &$0.00025$ &$0.012$ &$24$ 
&$0.073$ &$0.21$ &$0.28$ &$0.98$ \\ 
\hline 

\multirow{4}{*}{\parbox{1.5cm}{Partial differential equation}} 
& $30$ 
& $0.75$ & $230$ 
&N/A& OOT &{$2^{20}+1$} 
&$0.033$ &$5.6e-14$ &$0.033$ &$1500$ 
&$1$ &$5e-12$ &$1$ &$1.7$ \\ 
\cline{2-15} 
& $20$ 
& $0.038$ & $160$ 
&N/A& OOT &{$2^{20}+1$} 
&$0.033$ &$3.5e-14$ &$0.033$ &$890$ 
&$1$ &$5.4e-12$ &$1$ &$1.7$ \\ 
\cline{2-15} 
& $10$ 
& $0.086$ & $55$ 
&N/A& OOT &{$2^{20}+1$} 
&$0.033$ &$9.8e-13$ &$0.033$ &$520$ 
&$0.92$ &$2.7e-11$ &$0.92$ &$1.7$ \\ 
\cline{2-15} 
& $6$ 
& $0.1$ & $42$ 
&N/A& OOT &{$2^{20}+1$}  
&$0.033$ &$3.5e-07$ &$0.033$ &$370$ 
&$0.89$ &$5.5e-06$ &$0.89$ &$1.7$ \\ 
\hline 

\multirow{2}{*}{\parbox{1.5cm}{International space station}} 
& $25$ 
& N/A & OOT 
&N/A & OOT & {$2^{270}+3$} 
&N/A &$\begin{pmatrix}2.1e-05  \\ 0.001  \\ 4.6e-05 \end{pmatrix}$ & N/A & OOT 
&$\begin{pmatrix}0.00043  \\ 0.00026  \\ 0.00026 \end{pmatrix}$ &$\begin{pmatrix}0.47  \\ 0.47  \\ 0.47 \end{pmatrix}$ &$\begin{pmatrix}0.47  \\ 0.47  \\ 0.47 \end{pmatrix}$ & $11$  \\ 
\cline{2-15} 
& $10$ 
& N/A & OOT 
&N/A & OOT & {$2^{270}+3$} 
&N/A &$\begin{pmatrix}2.4e-05  \\ 5.6e-05  \\ 9e-05\end{pmatrix}$ & N/A & OOT  
&$\begin{pmatrix}0.00042  \\ 0.00022  \\ 0.00021 \end{pmatrix}$ &$\begin{pmatrix}1.7  \\ 1.7  \\ 1.7 \end{pmatrix}$ &$\begin{pmatrix}1.7  \\ 1.7  \\ 1.7 \end{pmatrix}$ & $12$  \\ 
\hline 

\multirow{3}{*}{\parbox{1.5cm}{FOM model}} 
& $20$ 
& N/A & OOT  
&N/A& OOT &{$2^{800}+1$} 
&N/A &$2.7e-07$ & N/A & OOT 
&$1.3$ &$1.1e-05$ &$1.3$ &$48$ \\ 
\cline{2-15} 
& $15$ 
& N/A & OOT 
&N/A& OOT &{$2^{800}+1$}  
&N/A &$0.00021$ & N/A & OOT 
&$1.3$ &$0.0065$ &$1.3$ &$48$ \\ 
\cline{2-15} 
& $10$ 
& N/A & OOT 
&N/A& OOT &{$2^{800}+1$} 
&N/A &$0.1$ &N/A & OOT 
&$1.3$ &$2.2$ &$3.5$ &$48$ \\ 
\hline 
\end{tabular}
} 
 \caption{The error bounds and computation times obtained from different methods on different benchmarks in which: $k$ is the dimension of the output abstraction, $\delta$ is total error bound, $e_1$ is the zero input response error, $e_2$ is the zero state response error, $t$ is the error computing time (in second) and $N$ is the number of simulations. The terms of ``N/A'' and ``OOT'' mean ``not applicable'' and ``out of time'', respectively.} 
\tablabel{experiment}
\end{table*} 

\begin{table*}[t] 

\resizebox{\textwidth}{!}{
\begin{tabular}{|p{2cm}|c|c|c|c|c|c|c|} 
\hline 
\multirow{2}{*}{\textbf{Benchmark}} & \multicolumn{2}{c|}{\textbf{Full Order System}} & \multicolumn{5}{c|}{\textbf{Output Abstraction}} \\ 
\cline{2-8} 
& \textbf{Time(s)} & \textbf{Memory(Kb)} & $k$ & $T_1(s)$ & $T_2(s)$ & \textbf{Total time(s)} & \textbf{Memory(Kb)} \\ 
\hline 
\multirow{2}{*}{\parbox{2cm}{Motor control system}} & \multirow{2}{*}{$27$} & \multirow{2}{*}{$3048$}
& $5$ &  $26$ & $0.92$ & $26.9$ &$3044$ \\
\cline{4-8}
&  &  & $4$ & $16.7$ & $0.9$ & $17.6$ & $3044$ \\
\hline

\multirow{3}{*}{\parbox{2cm}{Helicopter}} & \multirow{3}{*}{$287$} & \multirow{3}{*}{$3052$}
& $20$ & $206$ & $35$ & $241$ & $3052$  \\
\cline{4-8}
& & & $16$ & $128$ & $23$ & $151$ & $3048$ \\
\cline{4-8}
& & & $10$ & $68$ & $13$ & $81$ & $3048$ \\
\hline

\multirow{3}{*}{\parbox{2cm}{Building model}} & \multirow{3}{*}{$893$} & \multirow{3}{*}{$3056$}
& $25$ & $237.2$ & $130$ & $367.2$ & $3048$  \\
\cline{4-8}
& & & $15$ & $82.3$ & $58$ & $140.3$ & $3044$ \\
\cline{4-8}
& & & $6$ & $19.5$ & $24$ & $43.5$ & $3040$ \\
\hline

\multirow{3}{*}{\parbox{2cm}{Partial differential equation}} & \multirow{4}{*}{OOT} & \multirow{4}{*}{N/A}
& $30$ & $725.6$ & $1500$ & $2225.6$ & $3048$ \\
\cline{4-8}
& & & $20$ & $310$ & $890$ & $1200$ & $3048$ \\
\cline{4-8}
& & & $10$ & $75.2$ & $520$ & $595.2$ & $3040$  \\
\cline{4-8}
& & & $6$ & $31.9$ & $370$ & $401.9$ & $3040$  \\
\hline

\multirow{2}{*}{\parbox{2cm}{International space station}} & \multirow{2}{*}{OOT} & \multirow{2}{*}{N/A}
 & $25$ & $254.3$ & $11$ & $265.3$ & $3064$  \\
\cline{4-8}
& & & $10$ & $72.8$ & $12$ & $84.8$ & $3052$  \\
\hline

\multirow{3}{*}{\parbox{2cm}{FOM model}} & \multirow{3}{*}{OOT} & \multirow{3}{*}{N/A}
 & $20$ & $95.4$ & $48$ & $143.4$ & $3048$ \\
\cline{4-8}
& & & $15$ & $56.2$ & $48$ & $104.2$ & $3044$  \\
\cline{4-8}
& & & $10$ & $34.8$ & $48$& $82.8$ & $3040$ \\
\hline

\end{tabular} 

}
 \caption{Computation cost for verification process of the full order original LTI system and its output abstractions using SpaceEx~\cite{frehse2011spaceex} in which $T_1$ is the time for SpaceEx to compute the reach set of the output abstraction; $T_2$ is the time for obtaining the output abstraction; ``Total Time'' column states for the total time of verification process for the output abstraction, ``Memory'' column presents the memory used for computing reach set which is measure in kilobyte; time is measured in second. The terms of ``N/A'' and ``OOT'' mean ``not applicable'' and ``out of time''.} 
\tablabel{computation}
\end{table*}

To evaluate the order-reduction abstraction method presented in this paper, we implemented a software prototype that automatically creates output abstractions from full-order systems and applied it to a set of benchmarks. The method is integrated in HyST by calling Matlab related functions.\footnote{The prototype implementation and SpaceEx model files for the examples evaluated, both before and after order reduction, are available at: \url{http://verivital.com/hyst/pass-order-reduction/}.}
In this section, we first evaluate the advantages and disadvantages of our method in computing the error bound and its performance. Our results are compared with the results produced by the approximate bisimulation relations method~\cite{girard2007approximate} and the simulation-based approach~\cite{han2004reachability} via several benchmarks presented in~\tabref{benchmark}.
Then, we consider in detail how to apply our method to verify the safety of two specific case studies. 

\paragraph*{Error bound and computation time evaluation.~~}
The experiments are using Matlab 2014a and SpaceEx on a personal computer with the following configuration: Intel (R) Core(TM) i7-2677M CPU at 1.80GHz, 4GB RAM, and 64-bit Window 7. We set the upper limit for Matlab simulation and SpaceEx running time as two hours. It is said to be out of time (OOT) if we can not get the result after two hours.

\tabref{experiment} presents the error bounds and computation times of different methods on some typical benchmarks.
The approximate bisimulation method proposed in~\cite{girard2007approximate} is integrated in the Matlab toolbox called MATISSE.
The simulation-based method proposed in~\cite{han2004reachability} is done automatically in this paper.  
The results of our method are presented separately as follows. 
In the first part named ``Mixed bound'', we compute the bound of $e_1$ using~\thmref{e1_bound_opt} and the bound of $e_2$ using simulation. 
In the second part named ``Theoretical bound'', the bounds of $e_1$ and $e_2$ are computed using~\thmref{e1_bound} and~\thmref{e2_bound} respectively. 
We remind that, for $p$-output MIMO system, the simulation-based method and our proposed method compute separately the error bound for each pair of output (i.e. $\left\| y^i-y_r^i \right\|,~1\leq i \leq p$) while the approximate bisimulation method computes the total error bound (i.e. $\left\| y-y_r \right\|$). 

Let us consider the bound of $e_1$ related to the initial set of states $X_0$ that is computed using the two different techniques proposed in this paper.
For the helicopter and partial differential equation benchmarks, the initial set of states $X_0$ is far from the zero point.
The bounds of $e_1$ computed using~\thmref{e1_bound} is large and too conservative which may not be useful. 
For the motor control system and building model benchmarks, the initial set of states $X_0$ is close to the zero point. 
The bounds of $e_1$ computed by~\thmref{e1_bound} are fairly good and acceptable.
We can see that the bounds of $e_1$ computed using~\thmref{e1_bound_opt} is much smaller than the ones computed by~\thmref{e1_bound} for any situation of $X_0$.   

Now, we analyze the bound of $e_2$ computed by~\thmref{e2_bound} where the effect of Hankel singular values on this bound as mentioned in~\rmref{e2_rm} can be illustrated.
For the PDE benchmark, it can be seen that the theoretical bounds of $e_2$ for all cases of the output abstraction's dimension $k$ are very small due to the fact that the Hankel singular values $\sigma_k$ (which are not presented here) of the corresponding balanced system are very small (almost equal to zero) as $k \geq 5$. 
We can see more clearly the effect of these Hankel singular values by looking at the helicopter benchmark.   
The theoretical bound of $e_2$ becomes larger when the lower dimension output abstraction is obtained. 
It is small as $k$ equal to $20$ because $\sum_{21}^{28}\sigma_j$ is small. The theoretical bound of $e_2$ becomes conservative as $k = 10$ since $\sum_{11}^{28}\sigma_j$ is large. 
It can be shown that the bound of $e_2$ computed using simulation method is much less conservative than the theoretical bound.
Although~\thmref{e2_bound} may give conservative result for some systems, it is still useful for some other systems as analyzed above. 
The benefit of the theoretical bound is we can calculate the bound very quickly without doing simulation and thus avoid the numerical issues in simulation-based methods. 

We have discussed the benefits and drawbacks of different techniques proposed in this paper. 
Now, we make a short comparison with the approximate bisimulation relation method~\cite{girard2007approximate} and simulation-based method~\cite{han2004reachability}.
As can be seen from~\tabref{experiment}, the simulation-based approach gives very tight bounds for the errors (for examples, the motor control system and helicopter benchmarks). 
This approach is powerful when dealing with systems having small number of vertices in the initial set. 
When the number of vertices increases, the number of simulations also grows exponentially as can be seen from~\tabref{experiment}.
Therefore, it is difficult to apply the simulation-based approach in this situation (e.g, PDE, ISS and FOM benchmarks). For the approximate bisimulation relation method (integrated in Matisse toolbox), it can be observed that for PDE benchmark, this approach can give a good error bound. However, for MCS and Helicopter benchmarks, this approach gives very conservative results (which may not be useful) due to the appearance of ill-conditioned matrices in the process of solving LMI and optimization problems. We can see that for all the benchmarks on which the approximate bisimulation relation method can be applied, combination of using~\thmref{e1_bound_opt} and simulation bound of $e_2$ (i.e. mixed bound) produces much less conservative error bounds. When the dimension of the system is large (e.g, as in the ISS and FOM benchmarks), while the approximate bisimulation approach and~\thmref{e1_bound_opt} give no results due to running out of time, our theoretical approach can still be applied. 

Toward the computation time of different methods, we can see from the table that our method using~\thmref{e1_bound} and \thmref{e2_bound} has smallest computing time while using ~\thmref{e1_bound_opt} and approximate bisimulation relation method require much more time to compute the error bound.

In summary, we can use different methods to compute the error bound between the full-order system and the output abstraction.
Each method has benefits and drawbacks. The time complexity and the conservativeness of the result is a tradeoff that we need to take into account when applying these method to a specific system. As a suggestion from doing the experiment for this paper, for a system having dimension under $100$, we can generally use~\thmref{e1_bound_opt}, approximate bisimulation relation method~\cite{girard2007approximate} or simulation-based approach~\cite{han2004reachability} to compute the error bound. For systems with more than $100$ dimensions, we can use~\thmref{e1_bound} and~\thmref{e2_bound} or combine~\thmref{e1_bound} (for determining $e_1$ bound) and simulation-based approach (for computing $e_2$ bound).      

We have evaluated the error bounds and computation times of different methods. Next, we discuss about the benefit of using output abstraction for safety verification.~\tabref{computation} shows the computation cost of the verification process for the full-order LTI benchmarks and their different output abstractions. The bounded times for running all SpaceEx models are set as $t_f = 20s$. In the table, $T_1$ is the time for SpaceEx to compute the reach set of the output abstraction; $T_2$ is the time for obtaining the output abstraction; ``Total Time'' column states for the total time of verification process for the output abstraction, ``Memory'' column presents the memory used for computing reach set which is measure in kilobyte; all times are measure in second. For the first three systems (MCS, helicopter and BM), we combine~\thmref{e1_bound_opt} (for determining $e_1$ bound) and simulation-based approach (for computing $e_2$ bound) to derive the output abstraction. For the rest three benchmarks, we use~\thmref{e1_bound_opt}(for determining $e_1$ bound) and simulation-based approach (for computing $e_2$ bound) to obtain the output abstraction. As shown in the table, although using output abstraction does not help much to reduce the memory used in verification, it can help to reduce significantly the computation time. Moreover, output abstraction can be applied to check the safety of high-dimensional systems (e.g. PDE, ISS and FOM) that can not be verified directly using existing verification tools.       
Next, we consider the whole process of using output abstraction to verify the safety of two specific systems. 

\paragraph*{International Space Station (ISS).}
~The full-order model (denoted by $M_{270}$) of the component $R1$ of the international space station has $270$ state variables, three inputs and three outputs.
We refer reader to~\cite{Antoulas01asurvey} for the state space model of the system, and it is also included in our supplementary materials. The initial condition, input constraints and safety specification of the ISS system are presented in~\tabref{constraint}.

Verification for the full-order system with $270$ state variables may be difficult for existing verification tools.
Output abstraction and safety specification transformation can help to verify safety of such high-dimensional system with a small computation cost.
There are different output abstractions that can be used to verify whether the full-order system satisfies its safety requirements.
In this paper, we use a 10-order output abstraction and the corresponding transformed safety specification to check the safety of the full-order system.
From the safety requirement of the full-order system and the error bound shown in~\tabref{experiment}, we can see that the theoretical bound of $e_2$ is too conservative and cannot be used. 
To overcome this problem, we combine the theoretical bound of $e_1$ and the simulation bound of $e_2$ to derive a better bound between the full-order system and its 10-order output abstraction (denoted by $M^{\delta}_{10}$).
The error bound $\delta$ from this combination is $\delta = 10^{-3} \times [0.44,~0.28,~0.3]^T$.

The safety specification of the full-order ISS system $S_{270}$ is visualized by the region inside the middle blue polytopes in~\figref{iss_SP}. The transformed safety specifications (safe and unsafe specifications) of the corresponding 10-order output abstraction respectively are the region inside the smallest blue polytopes and the region outside the red polytopes. 


Figures~\ref{verify_S_12},~\ref{verify_S_23} and~\ref{verify_S_13} present the safety specification transformation and output reach set in the period of time $[0,~20s]$ computed by SpaceEx~\cite{frehse2011spaceex} of the $10$-order output abstraction on 2-dimension axes.

In the figures, the regions inside the middle blue polygons are the 2-dimensions projected safety regions of the full-order system.
The corresponding projected transformed safety and unsafe specifications $S_{10}^{\delta}$, $U_{10}^{\delta}$ of the output abstraction are described by the regions inside the smallest blue polygons and the regions outside the red polygons respectively.
The reach set $R_{ij}^{\delta},~i\neq j,~ 1\leq i,j \leq 3$ for each pair output $(y_{r_i},y_{r_j})$ of the abstraction $M_{10}^{\delta}$ are depicted by the solid blue regions.
As shown in the figures, for all $(i,j)$, we have $R_{ij}^{\delta}\cap \neg S_{10}^{\delta} = \emptyset$, or in other words, $M_{10}^{\delta} \vDash S_{10}^{\delta}$, thus it can be concluded that the full-order system $M_{270}$ satisfies the safety requirement $S_{270}$.
Therefore, the full-order system is safe.

\paragraph*{Periodically switched synchronous motor position control system.}
~We have applied our method for safety verification of a high-dimensional LTI system above.
Next, we consider how to use the proposed method to verify safety of a periodically switched synchronous motor system, which is used widely in many industrial fields such as elevator control systems, robotics and conveyer control systems.
In this system, two motors are controlled synchronously and periodically in both directions (i.e. clockwise and counterclockwise) to keep their position distance remaining in a desired range.

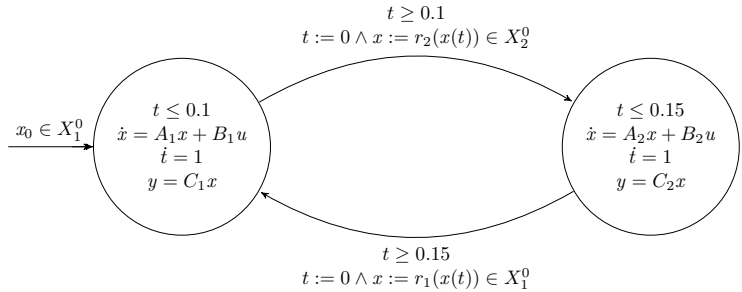
\begin{figure}[b]%
  \vspace{-1em}%
	\centering%
	\begin{adjustbox}{max size={0.8\columnwidth}{0.75\textheight}}%
	\begin{tikzpicture}[>=stealth',shorten >=1pt,auto,node distance=8cm,font=\normalsize]
		\tikzstyle{every state}=[minimum size=3cm,font=\normalsize]
		\node[state] (try)      			{\makecell[c]{$t \leq 0.15$\\$\dot{\mathit{x}} = A_2\mathit{x}+B_2\mathit{u}$\\$\dot{t} = 1$\\$\mathit{y} = C_2\mathit{x}$}};
			\node[initial,state,initial text={}]		 (rem) [left of=try]	 {\makecell[c]{$t \leq 0.1$\\$\dot{\mathit{x}} = A_1\mathit{x}+B_1\mathit{u}$\\$\dot{t} = 1$\\$\mathit{y} = C_1\mathit{x}$}};
			\path[->]			(rem)	edge[bend left] node{\makecell[c]{$t \geq 0.1$ \\ $ t := 0 \wedge \mathit{x} := r_2(\mathit{x}(t))\in X_2^0$}} (try);
			\path[->]			(try)	edge[bend left] node{\makecell[c]{$t \geq 0.15$ \\ $ t := 0 \wedge \mathit{x} := r_1(\mathit{x}(t))\in X_1^0$}} (rem);
			\draw[<-] (rem) -- node[above] {$\mathit{x}_0 \in X_1^0$} ++(-3cm,0);
	\end{tikzpicture}%
	\end{adjustbox}%
	\caption{Hybrid automaton model of the PSS synchronous motor control system.}%
	\figlabel{PSS_motor_sys}%
	\vspace{-1.25em}%
\end{figure}

Two motors have the same parameters with the motors used in Carnegie Mellon's undergraduate controls lab.
Each motor has its own controller which is designed to guarantee that: (a) the overshoot of the output does not exceed $16\%$; (b) the settling time is less than $0.04s$; (c) No steady-state error, even in the presence of a step disturbance input.
The system denoted by $M_8$ is modeled as a PSS with two modes as depicted in~\figref{PSS_motor_sys} in which two motors are controlled to rotate clockwise in mode 1 and inversely in mode 2.
The operating time in mode 1 is $t_1= 0.1$ and the operating time in mode 2 is $t_2 = 0.15$.

The system's matrices in the two modes are given by:
\begin{equation*}
\begin{split}
A_0 &= \begin{bmatrix} 0 & 1 & 0 & 0 \\ 0 & -1.0865 & 8487.2 & 0  \\ -2592.1 & -21.1190 & -698.9135 & -141390 \\ 1 & 0 & 0 & 0 \\ \end{bmatrix}, \\
B_0 &= \begin{bmatrix} 0 & 0 & 0 &-1 \\\end{bmatrix}^T, \\
A_1 &= A_2 = \begin{bmatrix} A_0 & 0 \\ 0 & A_0 \end{bmatrix}, B_1 = - B_2 = \begin{bmatrix} B_0 & 0 \\ 0 & B_0 \end{bmatrix},\\
C_1 &= C_2 = \begin{bmatrix} 1 & 0 & 0 & 0 & 0 & 0 & 0 & 0 \\ 1 & 0 & 0 & 0 & -1 & 0 & 0 & 0 \end{bmatrix}.
\end{split}
\end{equation*}

The reference control input applied to the system is $u = [u_1~u_2],~0.16 \leq u_1 \leq 0.2,~ 0.16 \leq u_2 \leq 0.22$.
The initial set of states of the system in the two modes are defined by the hyperbox:
\begin{equation*}
\begin{split}
&X^o_1 = \{x \in \mathbb{R}^{8} |~lb_1^i \leq x(i) \leq ub_1^i,  1 \leq i \leq 8 \}, \\
&X^o_2 = \{x \in \mathbb{R}^{8} |~lb_2^i \leq x(i) \leq ub_2^i,  1 \leq i \leq 8 \},
\end{split}
\end{equation*}
where $(lb_1,ub_1)$ and $(lb_2,ub_2)$ are initial conditions given in~\tabref{4}.

The first output of the system indicates the position of the first motor while the second output represents the position distance between the two motors.
In order to make the system operate safely, the two motors are controlled synchronously so that the first motor position $y_1$ and the position error between the two motors $y_2$ do not reach unsafe regions defined by $U(M_8) = \{(y_1,y_2) \in \mathbb{R}^2 |~ 178(y_1-0.325)^2 + 625(y_2-0.16)^2 \leq 1,~ 178(y_1+0.325)^2 + 625(y_2+0.16)^2 \leq 1 \}$.
The unsafe regions of the full-order system are visualized by the regions inside the smallest red ellipses in \figref{PSS_reach_set}.

To verify safety of the full-order ($8$-dimensional) system, we use a $5$th-order output abstraction $M_5^{\delta}$ and its transformed safety specification.

The matrices for the output abstraction denoted in mode 1 and mode 2 respectively are:
%
\begin{equation*}%
\begin{split}
A^{r}_1 &= A^{r}_2 = \begin{bmatrix} -18.925 & 80.823 & 0 & 0 & -29.973 \\ -80.823 & -76.569 & 0  &  0 & 122.93  \\ 0  & 0 & -18.925 & -80.823 & 0 \\  0  & 0 & 80.823 & -76.569  & 0 \\ -29.973 & -122.93  & 0 & 0 & -194.95 \\ \end{bmatrix}, \\
B^{r}_1 &= -B^{r}_2 = \begin{bmatrix} 5.7806 & 7.3762 & 2.2080 & -2.8175 & 4.8964 \\ -3.5726 & -4.5587 & 3.5726 & -4.5587  & -3.0262 \end{bmatrix}^T, \\
C^{r}_1 &= C^{r}_2 = \begin{bmatrix} 3.5726  & -4.5587  &  3.5726  &  4.5587  &  3.0262 \\  5.7806  & -7.3762  & -2.2080  & -2.8175  &  4.8964 \\ \end{bmatrix}.
\end{split}
\end{equation*}%

The transformed initial set of states the output abstraction in the two modes are defined by the hyperbox:
\begin{equation*}
\begin{split}
\widehat{X}^o_1 = \{x_r \in \mathbb{R}^{5} |~lb_{r1}^i \leq x_r^i \leq ub_{r1}^i,  1 \leq i \leq 5 \}, \\
\widehat{X}^o_2 = \{x_r \in \mathbb{R}^{5} |~lb_{r2}^i \leq x_r^i \leq ub_{r2}^i,  1 \leq i \leq 5 \},
\end{split}
\end{equation*}
where $(lb_{r1}^i, ub_{r1}^i)$ and $(lb_{r2}^i, ub_{r2}^i)$ are given in~\tabref{4}.

We combine the optimization (for $e_1$ bound) and simulation (for $e_2$ bound) methods to determine the error bounds between the full-order system and its 5-order output abstraction.
The error bounds in mode 1 and mode 2 respectively are $\delta_1 = [0.0234~~0.0189]^T$ and $\delta_2 = [0.0228~~0.0177]^T$. 
Using error bounds, the transformed unsafe specification for the output abstraction denoted by $U(M_5^{\delta})$ is: $U(M_5^{\delta}) =\{(y_1,y_2) \in \mathbb{R} |~ 178(y_1-0.325)^2 + 625(y_2-0.16)^2 \leq 1.57^2,~ 178(y_1+0.325)^2 + 625(y_2+0.16)^2 \leq 1.57^2 \}$.
The unsafe regions for the output abstraction are the regions inside the largest red ellipsoids in~\figref{PSS_reach_set}.

To ensure safety of the system, the output abstraction must not violate its transformed unsafe specification $U(M_5^{\delta})$. 
From~\figref{PSS_reach_set}, we can see that the output reach set of the output abstraction has an empty intersection with the unsafe regions, so we can conclude that the full-order system is safe.

\begin{figure}[t]
	\centering
		\includegraphics[width=0.8\columnwidth]{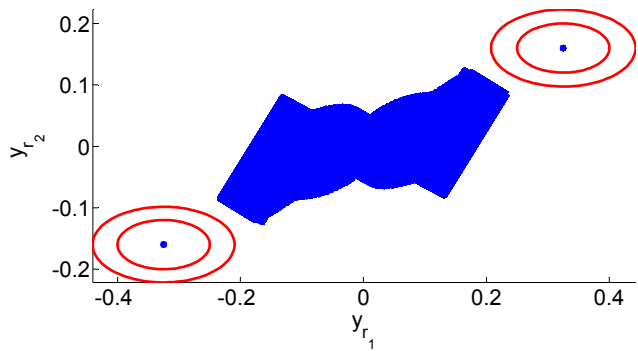}
	  \caption{Output reachable set in the period of time $[0,~20s]$ of the 5th-order output abstraction of the PSS synchronous motor position control system. The reach set does not reach the unsafe region (the region inside the largest red ellipse), thus the output abstraction is safe (with a bounded time interval), and thus the full-order system is safe (with a bounded time interval).}%
	  \figlabel{PSS_reach_set}%
		\vspace{-1em}%
\end{figure}

\begin{table}[t]%
	\tiny%
	\centering%
	\begin{tabular}{|l|l|}
  \hline
  Vector  & Value \\
	\hline
  $lb_1$  & [-0.002~~0~~0~~0~~-0.001~~0~~0~~0$]^T$  \\
  $ub_1$  & [0.0025~~0~~0~~0~~0.002~~0~~0~~0$]^T$  \\
	$lb_2$  & [-0.001~~0~~0~~0~~-0.002~~0~~0~~0$]^T$  \\
  $ub_2$  & [0.001~~0~~0~~0~~0.003~~0~~0~~0$]^T$  \\
  $lb_{r1}$  & [-0.1373e-03~~-0.5137e-03~~-0.0586e-03~~-0.2277e-03~~-0.2235e-03$]^T$  \\
  $ub_{r1}$  & [0.1323e-03~~0.5332e-03~~0.0930e-03~~0.3610e-03~~0.2320e-03$]^T$  \\
	$lb_{r2}$  & [-0.1211e-03~~-0.3684e-03~~-0.0687e-03~~-0.2666e-03~~-0.1603e-03$]^T$  \\
  $ub_{r2}$  & [0.0949e-03~~0.4703e-03~~0.0949e-03~~0.3684e-03~~0.2046e-03$]^T$  \\
  \hline
\end{tabular}
	\caption{Initial condition vectors of synchronous motor control system and its 5th-order output abstraction.}%
	\tablabel{4}
	\vspace{-2em}
\end{table}

\section{Conclusion and Future Work}
\seclabel{conclusion}
We have proposed an approach to verify safety specifications in high-dimensional linear systems and a class of periodically switched systems (PSSs) by verifying transformed safety specifications of a lower-dimensional output abstraction using existing hybrid system verification tools.
By reducing the dimensionality, our method significantly reduces the time and memory of reachability computations in the verification process.

There are several interesting directions for future work.
First, the method for calculating the error bound corresponding to the zero input response (i.e. $||e_1||$) can only be used for stable LTI systems.
Thus, a more general approach needs to be developed to deal with unstable linear systems.
It is also important to find a general strategy to address the verification problem for high-dimensional nonlinear systems.

Additionally, our approach can be extended to more general hybrid systems.
The main idea is that the states in each location that are related to guards/invariants need to be declared as the outputs of that location.
Then, the output abstraction for each location can be obtained.
A new hybrid system is then constructed based on these output abstractions.
The guards/invariants of the new hybrid system are obtained by transforming the former guards/invariants of the original hybrid system in the same manner of safety specifications transformation proposed in this paper.
This approach may benefit from other notions of ``similarity'' between behaviors (executions) of systems such as discrepancy functions~\cite{duggirala2013emsoftB}, or conformance degree~\cite{abbas2014formal}.

\bibliographystyle{spr-chicago} 
\bibliography{Tran_lib,deds2016}
\normalsize

\balance

\clearpage

\section{Appendix}
%
\subsection{Appendix: Proofs of Theorems}
\applabel{Appendix_alg}
In this appendix, we present proofs of theorems presented in this paper.

\subsubsection{Proof of~\thmref{e1_bound}}

The basic idea of determining the theoretical bound of the first error $e_1$ relies on the concept of monotonic convergence defined as follows. 

\begin{definition}
A homogeneous stable system $\dot{x} = Ax$ is called monotonic convergent if its states converge to zero and satisfy $\parallel x(t) \parallel \leq \parallel x(0) \parallel,~ \forall t \geq 0$.
\end{definition}

\begin{lemma}
\lemlabel{mono_convergence}
A homogeneous stable system $\dot{x} = Ax$ is a monotonic convergent system if $A+A^{T}<0$. 
\end{lemma}

\begin{proof}
Choose the Lyapunov function $V(x(t)) = x(t)^Tx(t)$, we have:
\begin{equation*}
\dot{V}(x(t))= x(t)^T(A+A^T)x(t)< 0.
\end{equation*}
Therefore, $V(x(t)) = \left\| x(t) \right\|^{2} \leq V(x(0))= \left\| x(0) \right\|^{2},~\forall t \geq 0$. This completes the proof.
\end{proof}

\textit{Proof of~\thmref{e1_bound}:}

We will show that the uncontrolled augmented system (i.e. $u = 0$) is a monotonic convergent system.
To proof above statement, let consider the uncontrolled-balanced system: $\dot{\tilde{x}} = \tilde{A}\tilde{x}$. 
Since the system is balanced, we have: 

\begin{equation*}
\begin{split}
\tilde{A}\Sigma+\Sigma\tilde{A}^{T}+BB^{T} &= 0\\
\tilde{A}^{T}\Sigma+\Sigma\tilde{A}+C^{T}C &= 0.
\end{split}
\end{equation*}

Combining two above equations yields: 

\begin{equation*}
(\tilde{A}+\tilde{A}^{T})\Sigma + \Sigma(\tilde{A}+\tilde{A}^{T}) = - BB^{T} - C^{T}C.
\end{equation*}

It is easy to see that the real parts of all eigenvalues of $\tilde{A}+\tilde{A}^{T}$ are necessarily non-positive. 
Since $\tilde{A}+\tilde{A}^{T}$ is symmetric, it is non-positive.
Note that $\tilde{A}$ is asymptotically stable.
Thus, using~\lemref{mono_convergence}, we can conclude that the uncontrolled-balanced system is a monotonic convergent system.

Similarly, we can see that the uncontrolled-reduced system $\dot{x}_r = A_rx_r$ is also a monotonic convergent system. 
Since $\tilde{A}+\tilde{A}^{T} < 0$ and $A_r + A_r^T < 0$, we have $\bar{A} + \bar{A}^T < 0$, that means the uncontrolled augmented system is a monotonic convergent system. 

Using the monotonic convergent property, the bound of the error $e_1$ satisfies: 

\begin{equation*}
\begin{split}
\left\| e_1^i(t) \right\|^2 & = \left\| \bar{y}(i) \right\|^2 = \bar{x}^T\bar{C}_i^T\bar{C}_i\bar{x} \\ 
                        &\leq \lambda_{max}(\bar{C}_i^T\bar{C}_i) \left\| \bar{x} \right\|^2   \\ 
												&\leq  \lambda_{max}(\bar{C}_i^T\bar{C}_i) \left\| \bar{x}_0 \right\|^2 \\ 											 &\leq  \lambda_{max}(\bar{C}_i^T\bar{C}_i) \cdot \sup_{x_0 \in X_0} \left\| \bar{x}_0 \right\|^2,~1 \leq i \leq p. 
\end{split}
\end{equation*}

This completes the proof.

\subsubsection{Proof of~\thmref{e1_bound_opt}}

Consider the uncontrolled augmented system (i.e. $u = 0$), let $V(\bar{x}(t)) = \bar{x}(t)^TP\bar{x}(t)$, we have $\dot{V}(\bar{x}(t)) = \bar{x}(t)^T(A^TP+PA)\bar{x}(t)$. 

Assume $P_0$ is the solution of the optimization problem in~\thmref{e1_bound_opt}. Because of $(A^TP_0+P_0A) < 0$, then $V(x(t)) < V(x(0) = \bar{x}_0^TP_0\bar{x}_0$. 
Note that $\left\| e_1^i(t) \right\|^2 = \bar{x}^T\bar{C}_i^T\bar{C}_i\bar{x},~ 1 \leq i \leq p$.
Since we also have $\bar{C}_i^T\bar{C}_i \leq P_0$, the bound of the error satisfies $\left\|e_1^i(t) \right\| \leq \sqrt{\bar{x}_{0}^T {P_0} \bar{x}_0}$

This completes the proof.

\subsubsection{Proof of~\thmref{e2_bound}}

The theoretical bound of the second error $e_2$ can be derived straightforwardly using the concept of bounded input bounded output stability and the $L_1$ error bound in impulse response of balanced truncation model reduction \cite{obinata2012model}. 
From \eqref{error}, we have: 

\begin{equation*}
\begin{split}
|e_2(t)| &= |\tilde{y}_u - y_{r_u}| = |\int_{0}^{t}(\tilde{C}e^{\tilde{A}(t-\tau)}\tilde{B} - C_{r}e^{A_{r}(t-\tau)}B_{r})u(\tau)d\tau | \\ 
      & \leq \int_{0}^{t}|(\tilde{C}e^{\tilde{A}(t-\tau)}\tilde{B} - C_{r}e^{A_{r}(t-\tau)}B_{r})||u| d\tau \\ 
			& \leq \left\| u \right\|_\infty \cdot \int_{0}^{\infty}|(\tilde{C}e^{\tilde{A}(t-\tau)}\tilde{B} - C_{r}e^{A_{r}(t-\tau)}B_{r})|d\tau \\
			& \leq \left\| u \right\|_\infty \cdot (2\sum_{j=k+1}^{n}(2j-1)\sigma_j).
\end{split}
\end{equation*}

Thus, $\left\| e_2^i(t) \right\| \leq \left\| u \right\|_\infty \cdot (2\sum_{j=k+1}^{n}(2j-1)\sigma_j)$ which completes the proof.

\subsubsection{Proof of~\lemref{3}}
%
%
From the definition of output abstraction, we have:
\begin{equation*}
\alpha_{ij}y_{r_j} - |\alpha_{ij}|\delta_j \leq \alpha_{ij}y_j \leq \alpha_{ij}y_{r_j} + |\alpha_{ij}|\delta_j.
\end{equation*}

\begin{equation*}
\Rightarrow \Gamma y_r + \overline{\Psi}_2 \leq \Gamma y + \Psi \leq \Gamma y_r + \overline{\Psi}_1.
\end{equation*}

Thus, $S(M_{k}^{\delta})$ and $U(M_{k}^{\delta})$ defined by~\eqref{SP_red_order_polytopes} satisfy the safety relation~\eqref{safety_relation}, which completes the proof.
%

\subsubsection{Proof of~\lemref{4}}

Let $\bar{y} = E(y-a)$, $\bar{y}_r = E(y_r - a)$.
We have:
\begin{equation} \eqlabel{proof_1}
\begin{split}
&(y-a)^TQ(y-a) = \bar{y}^T\Lambda\bar{y} = \sum_{i=1}^p\lambda_i\bar{y}_i^2, \\
&(y_r-a)^TQ(y_r-a) = \bar{y}_r^T\Lambda\bar{y}_r = \sum_{i=1}^p\lambda_i\bar{y}_{r_i}^2. \\
\end{split}
\end{equation}

From the definition of output abstraction (\defref{output_abstraction}), it is easy to see that:%
\begin{equation}\eqlabel{proof_2}
\begin{split}
-&\bar{\delta}_i \leq \bar{y}_i - \bar{y}_{r_i} = E(i,:)(y - y_r) \leq \bar{\delta}_i, \\
&\bar{\delta}_i = \sum_{j=1}^p|\gamma_{ij}|\delta_j. 
\end{split}
\end{equation}

Using~\eqref{proof_2} and the Cauchy-Schwarz inequality yields:
\begin{equation}\eqlabel{proof_3}
\begin{split}
&\sum_{i=1}^p\lambda_i(\bar{y}_i - \bar{y}_{r_i})^2 \leq \Delta_R^2 = \sum_{i=1}^p\lambda_i\bar{\delta}_i^2, \\ 
&\sum_{i=1}^p2\lambda_i\bar{y}_{r_i}(\bar{y}_i - \bar{y}_{r_i}) \leq 2\Delta_R\sqrt{\sum_{i=1}^p\lambda_i\bar{y}_{r_i}^2},\\
&\sum_{i=1}^p2\lambda_i\bar{y}_i(\bar{y}_{r_i} - \bar{y}_i) \leq 2\Delta_R\sqrt{\sum_{i=1}^p\lambda_i\bar{y}_i^2}.
\end{split}
\end{equation}

Combining the first and second inequality of~\eqref{proof_3} leads to:
\begin{equation}\eqlabel{proof_4}
\sum_{i=1}^p\lambda_i\bar{y}_i^2 \leq (\sqrt{\sum_{i=1}^p\lambda_i\bar{y}_{r_i}^2} + \Delta_R)^2. 
\end{equation}

Similarly, combining the first and the third inequality of~\eqref{proof_3} yields:
\begin{equation}\eqlabel{proof_5}
\sum_{i=1}^p\lambda_i\bar{y}_{r_i}^2 \leq (\sqrt{\sum_{i=1}^p\lambda_i\bar{y}_i^2} + \Delta_R)^2. 
\end{equation}

From~\eqref{proof_1},~\eqref{proof_4}, and~\eqref{proof_5}, we have:
\begin{equation}\eqlabel{proof_6}
\begin{split}
\sqrt{(y-a)^TQ(y-a)} \leq \sqrt{(y_r-a)^TQ(y_r-a)} + \Delta_R, \\
\sqrt{(y-a)^TQ(y-a)} \geq \sqrt{(y_r-a)^TQ(y_r-a)} - \Delta_R.
\end{split}
\end{equation}

Using~\eqref{proof_6}), we can conclude that $S(M_k^{\delta})$ and $S(M_k^{\delta})$ defined in~\lemref{4} 
satisfy the safety relation~\eqref{safety_relation}, which completes the proof.

\end{document}